\documentclass[runningheads]{llncs}
\usepackage{graphicx}
\usepackage{fullpage}
\usepackage{amsmath}
\usepackage[usenames,dvipsnames]{xcolor}
\usepackage[colorlinks,citecolor=blue,linkcolor=BrickRed]{hyperref}
\usepackage{makeidx} 
\usepackage{algorithm}
\usepackage{algorithmic}

\usepackage{tipa}
\usepackage{arcs,lmodern,fix-cm}
\usepackage{times}

\usepackage{amsfonts,latexsym,epsfig,amssymb,color}
\usepackage{mathdots,amsmath,amstext,setspace,enumerate}
\usepackage{slashbox,multirow}
\usepackage{rotating}
\usepackage{verbatim}
\usepackage{caption}
\usepackage{subcaption}

\usepackage{fourier}

\newcommand{\naturals}{\mathbb{N}}

\newcommand{\reals}{\mathbb{R}}

\newcommand{\ignore}[1]{}

\def\ex{\textrm{E}}
\newcommand{\search}[1]{${#1}$-\textsc{PDLS}}

\begin{document}
\title{Overcoming Probabilistic Faults in Disoriented Linear Search 
\thanks{Research supported in part by NSERC and by the Fields Institute for Research in Mathematical Sciences.}
\thanks{
This is the full version of the paper with the same title which will appear in the proceedings of the 
30th International Colloquium on Structural Information and Communication Complexity (SIROCCO’23), June 6-9, 2023, in Alcala de Henares, Spain.
}
}
%
%
\author{Konstantinos Georgiou\inst{1} 
\and Nikos Giachoudis\inst{1}
\and Evangelos Kranakis\inst{2}\orcidID{0000-0002-8959-4428}
}
\authorrunning{K. Georgiou et al.}
%
\institute{
Department of Mathematics, Toronto Metropolitan University, Toronto, ON, Canada \\
\email{konstantinos@torontomu.ca},~~\email{ngiachou@gmail.com}
\and
School of Computer Science, Carleton University, Ottawa, ON, Canada \\
\email{kranakis@scs.carleton.ca}
}

\maketitle              
\begin{abstract}
We consider search by mobile agents for a hidden, idle target, placed on the infinite line. Feasible solutions are agent trajectories in which all agents reach the target sooner or later. 
A special feature of our problem is that 
the agents are $p$-faulty, meaning that every attempt to change direction is an independent Bernoulli trial with known probability $p$, where $p$ is the probability that a turn fails. 
We are looking for agent trajectories that minimize the worst-case expected termination time, relative to the distance of the hidden target to the origin (competitive analysis).
Hence, searching with one $0$-faulty agent is the celebrated linear search (cow-path) problem that admits optimal $9$ and $4.59112$ competitive ratios, with deterministic and randomized algorithms, respectively. 


First, we study linear search with one deterministic $p$-faulty agent, i.e., with no access to random oracles, $p\in (0,1/2)$. For this problem, we provide trajectories that leverage the probabilistic faults into an algorithmic advantage. 
Our strongest result pertains to a search algorithm (deterministic, aside from the adversarial probabilistic faults) which, as $p\to 0$, has optimal performance $4.59112+\epsilon$, up to the additive term $\epsilon$ that can be arbitrarily small. Additionally, it has performance less than $9$ for $p\leq 0.390388$. When $p\to 1/2$, our algorithm has performance $\Theta(1/(1-2p))$, which we also show is optimal up to a constant factor.

Second, we consider linear search with two $p$-faulty agents, $p\in (0,1/2)$, for which we provide three algorithms of different advantages, all with a bounded competitive ratio even as $p\rightarrow 1/2$. 
Indeed, for this problem, we show how the agents can simulate the trajectory of any $0$-faulty agent (deterministic or randomized), independently of the underlying communication model. As a result, searching with two agents allows for a solution with a competitive ratio of $9+\epsilon$ (which we show can be achieved with arbitrarily high concentration) or a competitive ratio of $4.59112+\epsilon$. Our final contribution is a novel algorithm for searching with two $p$-faulty agents that achieves a competitive ratio $3+4\sqrt{p(1-p)}$, with arbitrarily high concentration. 



\end{abstract}
%
%
%




\section{Introduction}

Linear search refers to the problem of searching for a point target which has been placed at an unknown location on the real line. The searcher is a mobile agent that can move with maximum speed $1$ and is starting the search at the origin of the real line. The goal is to find the target in minimum time. This search problem provides a paradigm for understanding the limits of exploring the real line and has significant applications in mathematics and theoretical computer science. 

In the present paper we are interested in linear search under a faulty agent which is disoriented in that when it attempts to change direction not only it may fail to do so but also cannot recognize that the direction of movement has changed.
More precisely, for some $0 \leq p \leq 1$, a successful turn occurs with probability $1-p$ but the agent will not be able to recognize this until it has visited an anchor, a known, preassigned  point, placed on the real line. Despite this faulty behaviour of the agent it is rather surprising that it is possible to design algorithms which outperform the well-known zig-zag algorithm whose competitive ratio is $9$.

\subsection{Related Work}
Search by a single agent on the real line was initiated independently by Bellman \cite{bellman1963optimal} and Beck \cite{beck1964linear,beck1965more,beck1970yet} almost 50 years ago; the authors prove the well known result that a single searcher whose max speed is $1$ cannot find a hidden target placed at an initial distance $d$ from the searcher in time less than $9d$. These papers gave rise to numerous variants of linear search. Baeza-Yates et. al.~\cite{baezayates1993searching,baezayates1995parallel} study search problems by agents in other environments, e.g. in the plane or starting at the origin of $w$ concurrent rays (also known as the ``Lost Cow'' problem). Group search was initiated in~\cite{chrobak2015group} where evacuation (a problem similar to search but one minimizing the time it takes for the last agent to reach the target) by multiple agents that can communicate face-to-face was studied. An extension to the problem, where one tries to minimize the weighted average of the evacuation times was studied in~\cite{GEORGIOU20231}. 
There is extensive literature on this topic and \cite{czyzowicz2019groupkos} provides a brief survey of more recent topics on search.

Linear search with multiple agents some of which may be faulty, Crash or Byzantine, was initiated in the work of~\cite{czyzowicz2019search}~and~\cite{czyzowicz2021search}, respectively. In this theme, one uses the power of communication in order to overcome the presence of faults.
For three agents one of which is Byzantine, \cite{Sun2020} shows that the proportional schedule presented in~\cite{czyzowicz2019search} can be analyzed to achieve an upper bound of $8.653055$. Recently, \cite{czyzowicz2021searchnew} gives a new class of algorithms for $n$ agents when the number of Byzantine faulty among them is near majority, and the best known upper bound of $7.437011$ on an infinite line for three agents one of which is Byzantine. 

The present paper focuses on probabilistic search. The work of Bellman \cite{bellman1963optimal} and Beck \cite{beck1964linear,beck1965more,beck1970yet}, also mentioned above, has probabilistic focus. In addition numerous themes on probabilistic models of linear search can be found in the book~\cite{alpern2006theory} of search games and rendezvous, as well as in ~\cite{ahlswede1987search,stone1975theory}.

Search which takes into account the agent's turning cost is the focus of ~\cite{alpern2006theory}[Section 8.4] as well as the paper \cite{demaine2006online}. Search with uncertain detection is studied in~\cite{alpern2006theory}[Section 8.6]. According to this model the searcher is not sure to find the target when reaching it; instead it is assumed that the probability the searcher will find it on its $k$-th visit is $p_k$, where $\sum_{k \geq 0} p_k = 1$. A particular case of this is search with geometric detection probability~\cite{alpern2006theory}[Section 8.6.2] in which the probability of finding the target in the $k$-th visit is $(1-p)^{k-1}p$.  \cite{mccabe1974searching} investigates searching for a one-dimensional random walker and \cite{baston2001rendezvous} is concerned with rendezvous search when marks are left at the starting points. 
In another result pertaining to different kind of probabilistic faults,\cite{bonato2021probabilistically} studies the problem on the half-line (or $1$-ray), where detecting the target exhibits faults, i.e. every visitation is an independent Bernoulli trial with a known probability of success $p$. 
Back to searching the infinite line, a randomized algorithm with competitive ratio $4.59112$ for the cow path problem can be found in~\cite{kao1996searching} and is also shown to be optimal. 
In a strong sense, the results in this work are direct extensions of the optimal solutions for deterministic search in~\cite{baezayates1993searching} and for randomized search in~\cite{kao1996searching}. 
To the best of our knowledge the linear search problem considered in our paper has never been investigated before.  
We formally define our problem in Section~\ref{sec: model and problem definition}. 
Then in Section~\ref{sec: contrinutions} we elaborate further on the relevance of our results to~\cite{baezayates1993searching} and~\cite{kao1996searching}.

\section{Model \& Problem Definition (\search{p})}
\label{sec: model and problem definition}

We introduce and study the so-called Probabilistically Disoriented Linear Search problem (or \search{p}, for short), associated with some probability $p$. We generalize the well studied linear search problem (also known as cow-path) where the searcher's trajectory decisions exhibit probabilistic faults. The value $p$ will quantify a notion of probabilistic failure (disorientation).  

In \search{p}, an agent (searcher) can move at unit speed on an infinite line, where any change of direction does not incur extra cost. 
On the line there are two points, \emph{distinguishable} from any other point. Those points are the \textit{origin}, i.e. the agent's starting location, and the \textit{target}, which is what the agent is searching for and which can be detected when the agent walks over it. 

The agents have a faulty behaviour. 
If the agent tries to change direction (even after stopping), then with known  probability $p$ the agent will fail and she will still move towards the same direction. Consequent attempts to change direction are independent Bernoulli trials with probability of success $1-p$. Moreover the agent is \emph{oblivious} to the result of each Bernoulli trial, i.e. the agent is not aware if it manages to change direction. We think of this probabilistic behaviour as a co-routine of the agent that fails to be executed with probability $p$.
 An agent which satisfies this property for a given $p$ is called $p$-faulty. 
Moreover we assume, for the sake of simplicity, that at the very beginning the $p$-faulty agent starts moving to a specific direction, without fault. 

The agent's faulty behaviour is compensated by that it can utilize the origin and the target to recover its perception of orientation. Indeed, suppose that the agent passes over the origin and after time 1 it decides to change direction. In additional time 1, the agent has either reached the origin, in which case it realizes it turned successfully, or it does not see the origin, in which case it detects that it failed to turn. We elaborate more on this idea later. 

A solution to 
\search{p} is given by the agent's trajectory (or agents' trajectories), i.e. instructions for the agent(s) to turn at specific times which may depend on previous observations (e.g. visitations of the origin and when they occurred). A \emph{feasible trajectory} is a trajectory in which every point on the infinite line is visited (sooner or later) with probability 1 (hence a $1$-faulty agent admits no feasible trajectory). 

For a point $x\in \reals$ (target) on the line, we define the termination cost of a feasible trajectory in terms of competitive analysis. Indeed, if $E(x)$ denotes the expected time that target $x$ is reached by the last agent (so by the only agent, if searching with one agent), where the expectation is over the probabilistic faults or even over algorithmic randomized choices, then the termination cost for target (input) $x$ is defined as $E(x)/|x|$. The competitive ratio of the feasible trajectory is defined then as $
\limsup_x  E(x)/|x|$.

For the sake of simplicity, our definition deviates from the standard definition of the competitive ratio for linear search in which the performance is defined as the supremum over $x$ with absolute value bounded by a constant $d$, usually $d=1$.
However it can be easily seen that the two measures differ by at most $\epsilon$, for any $\epsilon>0$ using a standard re-scaling trick (see for example~\cite{GEORGIOU20231}) that shows why the value of $d$ is not important, rather what is important is that $d$ is known to the algorithm. 

\paragraph{Specifications when searching with two agents:}

When searching with two $p$-faulty agents, the value of $p$ is common to both, as well as the probabilistic faults they exhibit are assumed to be independent Bernoulli trials. 
The search by two $p$-faulty agents can be done either in the wireless or the face-to-face model. In the former model, we assume that agents are able to exchange messages instantaneously, whereas in the face-to-face model messages can be exchanged only when the agents are co-located. 

In either communication model, we assume that the two agents can detect that (and when) they meet. As a result, we naturally assume that upon a meeting, a $p$-faulty agent can also act as a \emph{distinguished point} (same as the origin), hence helping the other agent to turn. Later, we will call an agent who facilitates the turn \emph{Follower}, and the agent who performs the turn \emph{Leader}.
As long as the agents have a way to resolve the two roles (which will be built in to our algorithms), we also assume that the Leader moving in any direction can ``pick up'' another faulty agent she meets so that the two continue moving in that direction. 
 This means that two agents meeting at a point can continue moving to the direction of the leader (with probability $1$) even if the non-leader is idle. This property is motivated by that, effectively, the leader does not change direction, and hence there is no risk to make a mistake. 
Finally we note that two $p$-faulty agents can move at speed at most 1, independently of each other, as well as any of them can slow down or even stay put, complying still with the faulty turn specifications.

\subsection{Notes on Algorithmic and Adversarial Randomness }

The algorithms (feasible trajectories) that we consider are either deterministic or randomized, independently of the randomness induced by the faultiness. In particular, the efficiency measure is defined in the same way, where any expectations are calculated over the underlying probability space (induced by the combination of probabilistic faults and the possible randomized algorithmic choices). 
Moreover, if the algorithm is randomized then an additional random mechanism is used that is independent of the faulty behaviour. For example, a randomized agent (algorithm) can choose a number between $0$ and $1$ uniformly at random. 
A $p$-faulty agent that has access to a random oracle will be called \emph{randomized}, and \emph{deterministic} otherwise (but both exhibit probabilistically failed turns).

It follows by our definitions that \search{0} with one deterministic agent is the celebrated linear search problem (cow-path) which admits a provably optimal trajectory of competitive ratio 9~\cite{baezayates1993searching}. 
In the other extreme, \search{1} does not admit a feasible solution, since the agent moves indefinitely along one direction. In a similar spirit we show next in Lemma~\ref{lem: unbounded cr} that the problem is meaningful only when $p<1/2$, see Section~\ref{sec: proof of lem: unbounded cr} for proof. 


\begin{lemma}
\label{lem: unbounded cr}
No trajectory for \search{p} has bounded competitive ratio when $p\geq 1/2$. 
\end{lemma}

It is essential to note that in our model, the probabilistic faulty turns of a agent do hinder the control of the trajectory, but also introduce uncertainty of the algorithmic strategy for the adversary. 
As a result, the probabilistic movement of the agent (as long as $p>0$), even though it is not controlled by the algorithm, it can be interpreted as an algorithmic choice that is set to stone. 
Therefore, the negative result of~\cite{kao1996searching} implies the following lower bound for our problem.

\begin{corollary}
For any $p\in (0,1/2)$, no solution for \search{p} with one agent (deterministic or randomized) has competitive ratio lower than 4.59112.
\end{corollary}

\section{Contributions' Outline and Some Preliminary Results}
\label{sec: contrinutions}

Our main results in this work pertain to upper bounds for the competitive ratio that (one or two) faulty agents can achieve for \search{p}. 

\subsection{Results Outline for Searching with One Faulty Agent}

We start our exposition with search algorithms for one $p$-faulty agent. 
In Section~\ref{sec: deterministic algorithm} we analyze the performance of a deterministic search algorithm, whose performance is summarized in Theorem~\ref{thm: det algo limsup} on page~\pageref{thm: det algo limsup}. 
The section serves as a warm-up for the calculations we need for our main result when searching with one randomized $p$-faulty agent. Indeed, in Section~\ref{sec: randomized algorithm} we present a \emph{randomized} algorithm whose performance is summarized in Theorem~\ref{thm: rand algo limsup} on page~\pageref{thm: rand algo limsup}. The reader can see the involved formulas in the formal statements of the theorems, so here we summarize our results graphically in Figure~\ref{fig: summary of 1 agent}. Some important observations are in place. 

\begin{figure}
\begin{subfigure}[t]{0.45\textwidth}
\centering
\includegraphics[width =.9\textwidth ]{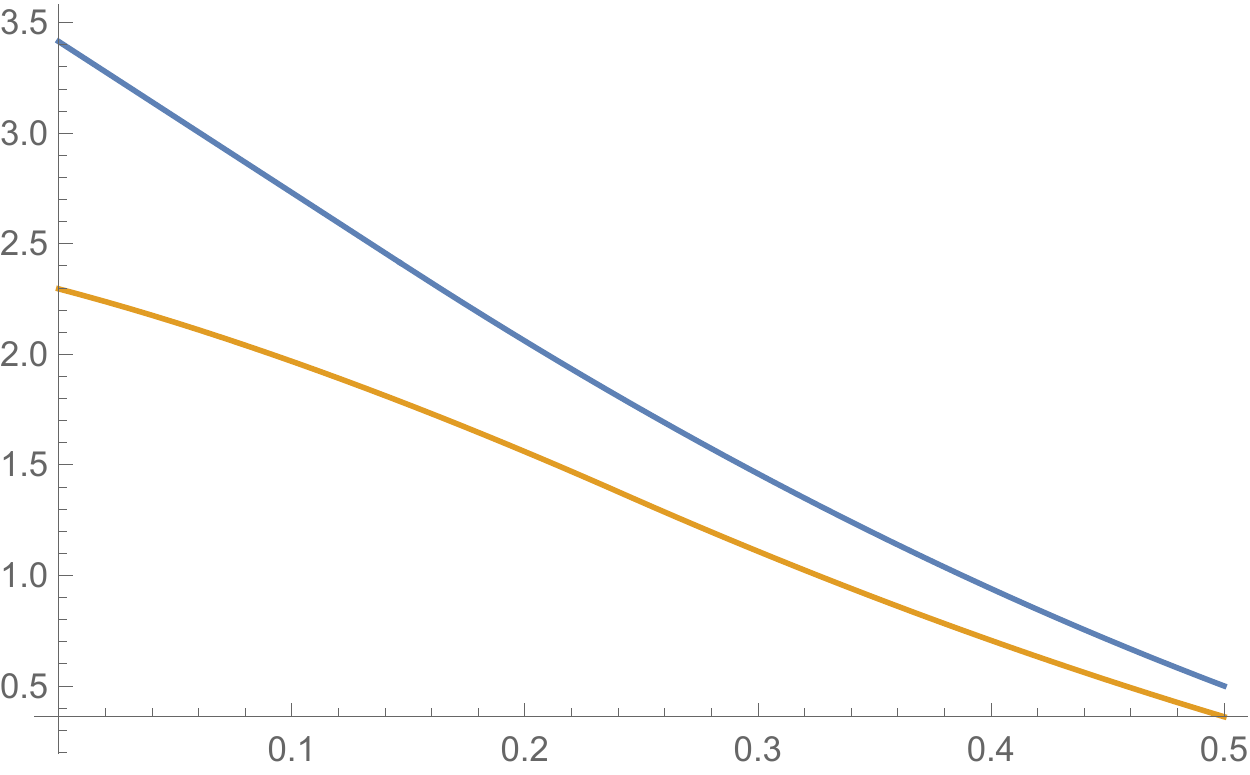}
\caption{
Competitive ratio comparison between the deterministic algorithm of Theorem~\ref{thm: det algo limsup} in blue and the randomized algorithm of Theorem~\ref{thm: rand algo limsup} in yellow. The competitive ratios are scaled by $(1/2-p)$.
}
\label{fig: DetRanCR}
\end{subfigure}\hfill
\begin{subfigure}[t]{0.45\textwidth }
\centering
\includegraphics[width =.9\textwidth]{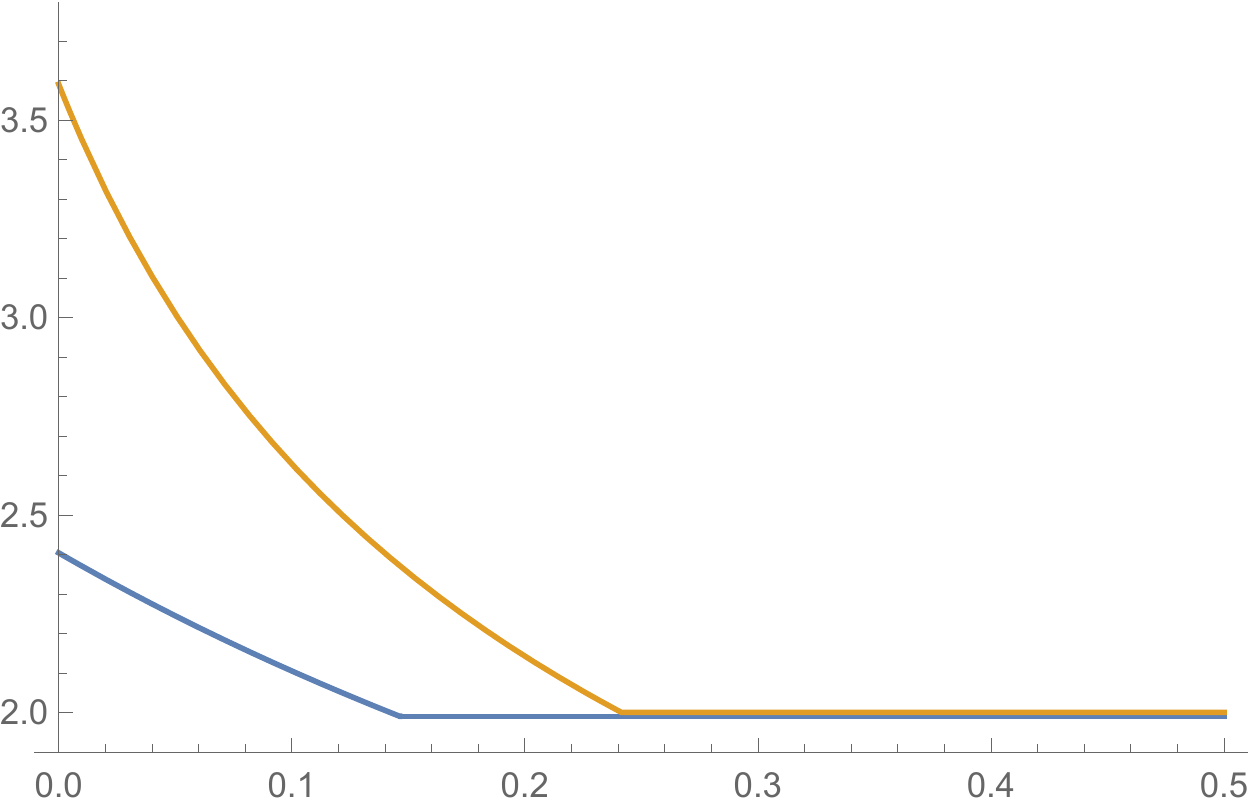}
\caption{
Intended expansion factors of the zig-zag algorithms of the deterministic algorithm of Theorem~\ref{thm: det algo limsup} in blue and of the randomized algorithm of Theorem~\ref{thm: rand algo limsup} in yellow. 
}
\label{fig: DetRanEF}
\end{subfigure}
\caption{Graphical summary of the positive results pertaining to searching with one $p$-faulty agent, $p\in (0,1/2)$.}
\label{fig: summary of 1 agent}
\end{figure}

\emph{Comments on the expansion factors:}
We emphasize that both our algorithms are adaptations of the standard zig-zag algorithms of~\cite{baezayates1993searching} and~\cite{kao1996searching} which are optimal for the deterministic and randomized model, respectively, when searching with a $0$-faulty agent. 
The zig-zag algorithms are parameterized by the so-called \emph{expansion factor} $g$ that quantifies the rate by which the searched space is increased in each iteration that the searcher changes direction. In our case however, we are searching with $p$-faulty agents, where $p>0$, and as a result, the algorithmic choice for how the searched space expands cannot be fully controlled (since turns are subject to probabilistic faults). This is the reason that, in our case, the analyses of these algorithms are highly technical, which is also one of our contributions. 

On a relevant note, the optimal expansion factor for the optimal deterministic $0$-faulty agent is 2, while the expansion factors $g=g(p)$ we use for the deterministic $p$-faulty agent, $p\in (0,1/2)$, are decreasing in $p$. As $p\rightarrow 0$, we use expansion factor $1+\sqrt{2}$, and the expansion factor drops to 2, for all 
$p\geq 0.146447$. 
When it comes to our randomized algorithm, the chosen expansion factor is again decreasing in $p$, starting from the same choice as for the optimal randomized $0$-faulty agent of~\cite{kao1996searching}, and being equal to $2$ for all 
$p\geq 0.241516$.\footnote{For simplicity, we only give numerical bounds on $p$. All mentioned bounds of $p$ in this section have explicit algebraic representations that will be discussed later.}
The expansion factors are depicted in Figure~\ref{fig: DetRanEF}.

\emph{Comments on the established competitive ratios:} 
By the proof of Lemma~\ref{lem: unbounded cr} it follows that as $p\to 1/2$, the optimal competitive ratio for \search{p}
is of order $\Omega(1/(1-2p)$. 
Hence for the sake of better exposition, we depict in Figure~\ref{fig: DetRanCR} the established competitive ratios scaled by $1/2-p$. Moreover, the results are optimal up to a constant factor when $p\to 1/2$. 

It is also interesting to note that for small enough values of $p$, the established competitive ratios are better than the celebrated optimal competitive ratio $9$
for \search{0}. 
This is because our algorithms leverage the probabilistic faults to their advantage, making the adversarial choices weaker. Indeed, algebraic calculations show that the competitive ratios of the deterministic and the randomized algorithms are less than 9 when 
$p\leq 0.390388$
and when
$p\leq 0.436185$,
respectively. 
Moreover, when searching with one $p$-faulty agent and $p\rightarrow 0$, 
the competitive ratio of our deterministic algorithm tends to $6.82843 <9$ and of our randomized algorithm to the provably optimal competitive ratio of $\frac{1}{W\left(\frac{1}{e}\right)}+1\approx 4.59112$, where $W(\cdot)$ is the  Lambert W-Function.
It is important to also note that for deterministic algorithms only, there is no continuity at $p=0$, since when the agent exhibits no faults, the adversary has certainty over the chosen trajectory and hence is strictly more powerful. 

Values of $p$ close to $1/2$ give rise to interesting observations too. Indeed, it is easy to see in Figure~\ref{fig: DetRanCR} that the derived comperitive ratios are of order
$\Theta\left( 1/(1-2p) \right)$.
More interestingly, the difference of the established competitive ratios of the deterministic and the randomized algorithms is $\Theta\left( 1/(1-2p) \right)$ too, when $p\rightarrow 1/2$. Hence the improvement when utilizing controlled (algorithmic) randomness is significant. We ask the following critical question: 
\begin{center}
\textit{``Does access to a random oracle provide an advantage for \search{p} with one agent?''}
\end{center}
Somehow surprisingly, we answer the question in the \emph{negative}!
More specifically, we show, for all $p \in (0,1/2)$ and using the probabilistic faults in our advantage, how a deterministic $p$-faulty agent (deterministic algorithm) can simulate a randomized $p$-faulty agent (randomized algorithm). Hence our improved upper bound of Theorem~\ref{thm: rand algo limsup} which is achieved by a randomized algorithm can be actually simulated by a deterministic $p$-faulty agent. The proof of this claim relies on that our randomized algorithm assumes access to a randomized oracle that samples only from uniform distributions a finite number of times (in fact only 2). 
The main idea is that a deterministic agent can stay arbitrarily close to the origin, sampling arbitrarily many random bits using it's faulty turns, allowing her to simulate queries to a random oracle. 
The details of the proof of the next theorem appear in Section~\ref{sec: proof of thm: deterministic simulates randomized}.

\begin{theorem}
\label{thm: deterministic simulates randomized}
For any $p \in (0,1/2)$, let $c$ be the competitive ratio achieved by a randomized faulty agent for \search{p}, having access (finite many times) to an oracle sampling from the uniform distribution. Then for every $\epsilon >0$, there is a deterministic $p$-faulty agent for the same problem with competitive ratio at most $c+\epsilon$. 
\end{theorem}

\subsection{Results Outline for Searching with Two Faulty Agents}

We conclude our contributions in Section~\ref{sec: two agents} where we study 
\search{p} with two agents. 
First we show how two (deterministic) $p$-faulty agents (independently of the underlying communication model) can simulate the trajectory of any (one) $0$-faulty agent. 
As an immediate corollary, we derive in Theorem~\ref{thm: 2 agents 9 competitive} on page~\pageref{thm: 2 agents 9 competitive} a method for finding the target with two $p$-faulty agents that has competitive ratio $9+\epsilon$, for every $\epsilon>0$. Most importantly as long as $p>0$, the result holds not only in expectation, but with arbitrarily large concentration. 

Motivated by similar ideas, we show in Theorem~\ref{thm: 2 agents 4.59112 competitive} on page~\pageref{thm: 2 agents 4.59112 competitive} how two deterministic $p$-faulty agents can simulate the celebrated optimal randomized algorithm for one agent for \search{0}, achieving competitive ratio arbitrarily close to $4.59112$.
The result holds again regardless of the communication model. 
However, in this case we cannot guarantee a similar concentration property as before. 

Finally, we study the problem of searching with two \emph{wireless} $p$-faulty agents. Here we are able to show in Theorem~\ref{thm: 2 agents improved competitive} on page~\pageref{thm: 2 agents improved competitive} how both  agents can reach any target with competitive ratio $3+4\sqrt{p(1-p)}$, in expectation. The performance is increasing in $p<1/2$, ranging from $3$ (the optimal competitive ratio when searching with two wireless $0$-fault agents) to $5$.
However, as before we can control the concentration of the performance, making it arbitrarily close to $3+4\sqrt{p(1-p)}$ with arbitrary confidence. Hence, this gives an advantage over Theorem~\ref{thm: 2 agents 4.59112 competitive} for small values of $p$, i.e. when the derived competitive ratio is smaller than $4.59112$. In this direction, we show that when $p> 0.197063$ the competitive ratio exceeds $4.59112$, and hence each of the results described above are powerful in their own right.

\section{Searching with One Deterministic $p$-Faulty Agent} 
\label{sec: deterministic algorithm}

We start by describing a deterministic algorithm for searching with a $p$-faulty agent, where $p\in (0,1/2)$. Our main result in this section reads as follows. 
\begin{theorem}
\label{thm: det algo limsup}
\search{p} with one agent 
admits a deterministic algorithm with competitive ratio equal to $2(\sqrt2+2)$ when $0<p\leq \frac14(2-\sqrt2)$, and equal to $6-4p+\frac{1}{1-2p}$ when $\frac14(2-\sqrt2) <p<1/2$.
\end{theorem}

We emphasize that having $p>0$ will be essential in our algorithm. This is because the probabilistic faults introduce uncertainty for the adversary. For this reason, it is interesting but not surprising that we can in fact have competitive ratio $2(\sqrt2+2)\approx 6.82843 <9$ for all $p<\frac{1}{8} \left(\sqrt{17}-1\right)\approx 0.390388$. 

First we give a verbose description of our algorithm, that takes as input $p\in(0,1)$, and chooses parameter $g=g(p)$ which will be the intended expansion rate of the searched space. 
In each iteration of the algorithm, the agent will be passing from the origin with the intention to expand up to $g^i$ in a certain direction. When distance $g^i$ is covered, the agent attempts to return to the origin. After additional time $g^i$ the agent knows if the origin is reached, in which case she expands in the opposite direction with intended expansion $g^{i+1}$. If not, the agent knows she did not manage to turn, and proceeds up to point $g^{i+1}$ (this is why we require that $g\geq 2$). Then, she makes another attempt to turn. This continues up to the event that the agent succeeds in turning at $g^j$, for some $j>i$, and then the agent attempts to expand in the opposite direction with intended expansion length $g^{j+1}$. 
The algorithm starts by searching towards an arbitrary direction, say right, with intended expansion $g^0$. 
Llater on, it will become clear that the termination time of the algorithm converges only if $g<1/p$. 

In order to simplify the exposition (and avoid repetitions), we introduce first subroutine 
Algorithm~\ref{algo: baseline search} which is the baseline of both algorithms we present for searching with one agent. Here, this is followed by Algorithm~\ref{algo: deterministic} which is our first search algorithm. 

\begin{algorithm}
    \caption{Baseline Search} 
    \label{algo: baseline search}
    \begin{algorithmic}[1]
        \REQUIRE $g,p,i,d$
        \REPEAT
            \REPEAT
                \STATE Move towards point $dg^i$ with speed $1$
            \UNTIL{The target is found $\vee$ $g^i$ time has passed}
            \WHILE{The origin is $\neg$found $\wedge$ the target is $\neg$found}
                \STATE Set $d \leftarrow -d$ \COMMENT{This changes the direction but it can fail with probability $p$}
                \STATE Set $i \leftarrow i+1$
                \REPEAT
                    \STATE Move towards point $dg^i$ with speed $1$
                \UNTIL{Either the target is found or $g^i-g^{i-1}$ time has passed or the origin is found}
            \ENDWHILE
        \UNTIL{The target is found}
    \end{algorithmic}
\end{algorithm}

\begin{algorithm}
    \caption{Faulty Deterministic Linear Search} 
    \label{algo: deterministic}
    \begin{algorithmic}[1]
        \REQUIRE $2\leq g \leq 1/p$ and $0 < p < 1/2$
        \STATE Set $i \leftarrow 0$ and $d \leftarrow 1$ \COMMENT{$d\in\{1,-1\}$ represents direction ($1$ going right and $-1$ going left)}
	\STATE Run Algorithm~\ref{algo: baseline search} with parameters $g,p,i,d$.
    \end{algorithmic}
\end{algorithm}

As we explain momentarily, Theorem~\ref{thm: det algo limsup} follows directly by the following technical lemma, whose proof appears in Section~\ref{sec: proof of lem: cr deterministic in interval}. 

\begin{lemma}
\label{lem: cr deterministic in interval}
Fix $p\in (0,1/2)$ and $g\in [2,1/p)$. If the target is placed in $(g^t, g^{t+1}]$, 
then the competitive ratio of Algorithm~\ref{algo: deterministic} is at most 
$$
\frac{
(1 - p) g 
}
{(1-g p)}
\left(
\frac{
1+
(1-2 p) g
}
{
g-1
}
+
\left(1-2 p\right)^{t+1}
\right)
+1.
$$
\end{lemma}

Taking the limit when $t\to \infty$ of the expression of Lemma~\ref{lem: cr deterministic in interval} shows that the competitive ratio of Algorithm~\ref{algo: deterministic} is 
$$
f_p^{\textsc{det}}(g):=\frac{1-g (2 g (p-2) p+g+2)}{(1-g) (1-g p)}
$$
for all $p,g$ complying with the premise.\footnote{If one wants to use the original definition of the competitive ratio, then by properly re-scaling the searched space (just by scaling the intended turning points), one can achieve competitive ratio which is additively off by at most $\epsilon$ from the achieved value, for any $\epsilon>0$.} 

Next we optimize function $f_p^{\textsc{det}}(g)$. When $p \leq \frac{1}{4} \left(2-\sqrt{2}\right)$, the optimal expansion factor (optimizer of $f_p^{\textsc{det}}(g)$) is 
$$
g_0(p):= \frac{-\left(\sqrt{2}+2\right) p+\sqrt{2}+1}{1-2 p^2}.
$$
It is easy to see that $2\leq g_0(p) <1/p$, for all $p \leq \frac{1}{4} \left(2-\sqrt{2}\right)$ (in fact the strict inequality holds for all $0<p<1$). In this case the induced competitive ratio becomes  $2 \left(\sqrt{2}+2\right)\approx 6.82843$. 

When $p\geq \frac{1}{4} \left(2-\sqrt{2}\right)$, the optimal expansion factor (at least $2$) is $g=2$, in which case the competitive ratio becomes
$
6-4 p+\frac{1}{1-2 p}.
$
Interestingly, the competitive ratio becomes at least $9$ for $p\geq \frac{1}{8} \left(\sqrt{17}-1\right) \approx  0.390388$. In other words, $0.390388$ is a threshold for the probability associated with the agent's faultiness for which, at least for the proposed algorithm, that probabilistic faultiness is useful anymore towards beating the provably optimal deterministic bound of $9$. On a relevant note, recall that by Lemma~\ref{lem: unbounded cr}, there is a threshold probability $p'$ such that any algorithm (even randomized) has competitive ratio at least $9$ when searching with a $p$-faulty agent, when $p\geq p'$.

\section{Searching with One Randomized (Improved Deterministic) $p$-Faulty Agent}
\label{sec: randomized algorithm}

In this section we equip the $p$-faulty agent with the power of randomness, and we show the next positive result. Note that due to Theorem~\ref{thm: deterministic simulates randomized}, the results can be simulated by a deterministic $p$-faulty agent too, up to any precision.

\begin{theorem}
\label{thm: rand algo limsup}
Let $p_0 = \frac{1}{8} \left(5-\sqrt{1+\ln ^2(2)+6 \ln (2)}-\ln (2)\right) \approx 0.241516$. 
For each $p <p_0$, let also $g_p$ denote the unique root, no less than 2, of
\begin{equation}
\label{equa: derivative f rand}
h_p(g):= (1-g p) (1+g (1-2 p) ) - g (1-p) \ln g. 
\end{equation}

Then for each $p \in (0,1/2)$, \search{p} with one agent admits a randomized algorithm with competitive ratio at most
$$
\left\{
\begin{array}{ll}
g_p \left( \frac{1-p}{1-g_pp} \right)^2+1, & p \in (0,p_0] \\
(1-p)\frac{ 1+2 (1-2 p)}{(1-2 p) \ln (2)}+1, & p \in (p_0,1/2).
\end{array}
\right.
$$
\end{theorem}
Elementary calculations show that the competitive ratio of Theorem~\ref{thm: rand algo limsup} remains at most 9, as long as 
$$p\leq \frac{1}{8} \left(7+\sqrt{1+256 \ln ^2(2)-96 \ln (2)}-16 \ln (2)\right) \approx 0.436185.$$

First we argue that the premise regarding $g_p$ of Theorem~\ref{thm: rand algo limsup} is well defined. Indeed, consider $h_p(g)$
as in~\eqref{equa: derivative f rand}. 
We have that $\partial h_p(g) / \partial p = 4 g^2 p-g^2-3 g+g \ln g \leq -3 g + g \ln g <0$, for all $p \leq 1/4$. 
Therefore, $h_p(2) = 8 p^2-10 p+2 p \ln (2)+3-2$ is decreasing in $p$, and hence $h_p(2)>h_{p_0}(2) = 0$, by the definition of $p_0$. Also
$h_p(4) = 32 p^2-28 p+4 p \ln (4)+5-4 \ln 4$ is decreasing in $p$ too, and therefore
$h_p(4) \leq h_0(4) = 5-4 \ln 4 \approx -0.545177 <0$.
This means that $h_p(g)$ has indeed a root, with respect to $g$, in $[2,4)$, for all $p \in (0,p_0]$. Next we show that this root is unique. Indeed, we have 
$\partial h_p(g) / \partial g = 4 g p^2-2 g p+p \ln g-\ln g-2 p \leq - \ln g <0$, for $p \in [0,1/2]$. 
Therefore, for all $p\leq p_0<1/2$, function $h_p(g)$ is decreasing in $g$, and hence any root is unique, that is $g_p$ is indeed well-defined and lies in the interval $[2,1/p)$, for all $p \in (0,1/2)$.

Next, we observe that as $p\rightarrow 0$, the competitive ratio promised by Theorem~\ref{thm: rand algo limsup} when searching with a $p$-faulty agent (i.e. a nearly non-faulty agent) is given by $g_0$ being the root of $g-g \ln g+1$, i.e. $g_0 = \frac{1}{W\left(\frac{1}{e}\right)} \approx 3.59112$, where $W(\cdot)$ is the Lambert W-Function.\footnote{The Lambert W-Function is the inverse function of $L(x)=x e^x$.} Moreover, the induced competitive ratio is $1+g_0 \approx 4.59112$, which is exactly the competitive ratio of the optimal randomized algorithm for the original linear search problem due to Kao et al~\cite{kao1996searching}. We view this also as a sanity check regarding the correctness of our calculations.

Not surprisingly, our algorithm that proves Theorem~\ref{thm: rand algo limsup} is an adaptation of the celebrated optimal randomized algorithm for linear search with a $0$-faulty agent of~\cite{kao1996searching}.
As before, the search algorithm, see Algorithm~\ref{algo: randomized} below, is determined by some expansion factor $g$, that represents the (intended, in our case, due to faults) factor by which the searched space is extended to, each time the direction of movement changes. 
The randomized algorithm makes two queries to a random oracle. 
First it chooses a random bit, representing an initial random direction to follow. 
Second, the algorithm samples random variable $\epsilon$ that takes a value in $[0,1]$, uniformly at random. Variable $\epsilon$ quantifies a random scale to the intended turning points.
It is interesting to note that setting $\epsilon = 0$ in Algorithm~\ref{algo: randomized} deterministically, and removing the initial random direction choice, gives rise to the previous deterministic Algorithm~\ref{algo: deterministic}.
 


\begin{algorithm}
    \caption{Faulty Randomized Linear Search} 
    \label{algo: randomized}
    \begin{algorithmic}[1]
        \REQUIRE $g \geq 2$ and $0 < p < 1/2$
        \STATE Choose $\epsilon$ from $[0,1)$ uniformly at random
        \STATE Set $i \leftarrow \epsilon$ and $d \leftarrow 1$ \COMMENT{$d\in\{1,-1\}$ represents direction ($1$ going right and $-1$ going left)}
	\STATE Run Algorithm~\ref{algo: baseline search} with parameters $g,p,i,d$.
    \end{algorithmic}
\end{algorithm}

We show next how Theorem~\ref{thm: rand algo limsup} follows by the following technical lemma, whose proof appears in Section~\ref{sec: proof of lem: cr randomized in interval}. 

\begin{lemma}
\label{lem: cr randomized in interval}
For any $p\in (0,1/2)$ and any $g\in [2,1/p)$, if the target is placed in the interval $(g^t, g^{t+1}]$, then the competitive ratio of Algorithm~\ref{algo: randomized} is at most
$$
    1 + 
\frac{1}{(1-gp) \ln g}
\left(
\begin{array}{ll}
(1 - p) (1 - g (-1 + 2 p)) (1 + (-1 + 2 p)^t)  \\
~~~ + \frac{2 (1 - p)}{g^t} 
 +2 g (1 - p) (p + g^t (-1 + p) (-1 + 2 p)^t
\end{array}
\right).
$$
\end{lemma}

Recall that $p \in (0,1/2)$, so taking the limit of the expression of Lemma~\ref{lem: cr randomized in interval} when $t\rightarrow \infty$, and after simplifying algebraically the expression, shows that the competitive ratio of Algorithm~\ref{algo: randomized} is at most 
$$
f_p^{\textsc{rand}}(g):=    1 + \frac{(1 - p) (1 + g (1 - 2 p))}{(1 - g p) \ln g},
$$
for all $p,g$ complying with the premise. Next we optimize $f_p^{\textsc{rand}}(g)$ for all parameters $p\in (0,1/2)$, under the constraint that $2\leq g <1/p$. 
In particular, we show that the optimizers of $f_p^{\textsc{rand}}(g)$ are $g_p$, if $p\leq p_0$, and $g=2$ otherwise resulting in the competitive ratios as described in Theorem~\ref{thm: rand algo limsup}.

First we compute
$
\frac{\partial  }{\partial g } f_p^{\textsc{rand}}(g)
=
\frac{1-p }{g (1-g p)^2 \ln ^2g} h_p(g),
$
where $h_p(g)$ is the same as~\eqref{equa: derivative f rand}. As already proven below the statement of Theorem~\ref{thm: rand algo limsup}, we have that $h_p(g)$ has a unique root in the interval $[2,4)$. 
Next we show that $f_p^{\textsc{rand}}(g)$ is convex.

Indeed, we have that 
$$
\frac{\partial^2  }{\partial g^2 } f_p^{\textsc{rand}}(g)
=
\frac{1-p}{g^2 (1-g p)^3 \ln ^3(g)} s_p(g),
$$
where $s_p(g)=\sum_{i=0}^3\alpha_i p^i$ is a degree 3 polynomial in $p$ with coefficients 
\begin{align*}
\alpha_3&=2 g+g (-\ln g)+\ln (g)+2 \\
\alpha_2&=2 g \left(-2 g+g \ln ^2g-\ln g-4\right) \\
\alpha_1&= g^2 \left(2 g-2 \ln ^2g+g \ln g+3 \ln g+10\right) \\
\alpha_0&=-2 g^3 (\ln g+2) 
\end{align*}
As a result, it is easy to verify that $s_p(g)$ remains positive for all $p \in (0,1/2)$, condition on that $g\in [2,4]$ (in fact the optimizers as described in Theorem~\ref{thm: rand algo limsup} do satisfy this property). 
We conclude that $f_p^{\textsc{rand}}(g)$ is convex in $g$.

Together with our previous observation, this means that, under constraint $g\geq 2$, function $f_p^{\textsc{rand}}(g)$ is minimized at the unique root of $\frac{\partial  }{\partial g } f_p^{\textsc{rand}}(g)$ when $p\leq p_0$, and at $g=2$ when $p \in [p_0,1/2)$. These are exactly the optimizers described in Theorem~\ref{thm: rand algo limsup}, where in particular the competitive ratio $f_p^{\textsc{rand}}(g)$ is simplified taking into consideration that for the chosen value of $g$ we have that $s_p(g)=0$, for all $p\leq p_0$.

Lastly, it remains to argue that all optimizers of $f_p^{\textsc{rand}}(g)$ are indeed at most $1/p$. For this, it is enough to show that the unique root $g_p$ of $h_p(g)$ is at most $1/p$, for all $p\leq p_0$ (since for larger values of $p$ we use expansion factor $g=2$). For this, and since $g\geq 2$, we have 
$
h_p(g) \leq 2 g^2 p^2-g^2 p-3 g p+g p \ln 2+g-g \ln 2+1.
$
The latest expression is a polynomial in $p$ of degree 2, which has only one of its roots positive, namely
$$
\frac{-\sqrt{(-3 p+p \ln 2+1-\ln 2)^2-4 \left(2 p^2-p\right)}+3 p+p (-\ln 2)-1+\ln 2}
{2p \left(2 p-1\right)}.
$$
Simple calculations then can show that the latter expression is at most $1/p$ for all $p\in (0,1/2)$. In fact one can show that the expression above, multiplied by $p$, is strictly increasing in $p$ and at $p=1/4 >p_0$ becomes 
$$\frac{1}{4} \left(1-\ln8+\sqrt{9+(\ln8-2) \ln8}\right) \approx 0.486991 <1/2.$$
That shows that for each $p<p_0$ we have that for the unique root $g_p$ of $h_p(g)$ the inequality $g_p<1/2p<1/p$ is valid, as desired.

\section{Searching with Two $p$-Faulty Agents}
\label{sec: two agents}

In this section we present algorithms for \search{p} for two faulty agents, for all $p\in (0,1/2)$. 
Central to our initial results is the following subroutine that, at a high level, will be used by a $p$-faulty agent, the \emph{Leader}, in order to make a ``forced'' turn, with the help of a \emph{Follower}, which can be either a distinguished immobile point, e.g. the target or the origin, or another $p$-faulty agent. 
In this process the $p$-faulty agent who undertakes the role of the Follower may need to either slow down or even halt for some time, still complying with the probabilistic faulty turns (once halted, she can continue moving in the previous direction, but changing it is subject to a fault). 

\begin{algorithm}
    \caption{Force Change Direction (Instructions for a Leader)} 
    \label{algo:subroutine change direction}
    \begin{algorithmic}[1]
        \REQUIRE $\gamma$ small real number, Follower either mobile or immobile. 
            \REPEAT
                \STATE Change direction \COMMENT{This fails with probability $p$}
                \STATE Move for time $\gamma$ if Follower is mobile, and $\gamma \leftarrow 2\gamma$ if Follower is immobile.
            \UNTIL{You meet with follower}
            \STATE Communicate to Follower the Leader's direction 
    \end{algorithmic}
\end{algorithm}

It will be evident, in the proof of Lemma~\ref{lem: 2 agents simulate 1} below, that Algorithm~\ref{algo:subroutine change direction} will allow an agent to change direction arbitrarily close to an intended turning point, and with arbitrary concentration (both controlled by parameter $\gamma$). 

The next lemma refers to a task that two $p$-faulty agents can accomplish independently of the underlying communication model. At a high level, the lemma establishes that two $p$-faulty agents can bypass the probabilistic faults at the expense of giving up the independence of the searchers' moves. 

\begin{lemma}
\label{lem: 2 agents simulate 1}
For every $p\in [0,1/2)$, two $p$-faulty agents can simulate the trajectory of a deterministic $0$-faulty agent within any precision (and any probability concentration).
\end{lemma}

\begin{proof}
Consider two $p$-faulty agents that are initially collocated at the origin. We show how the agents can simulate (at any precision) a deterministic trajectory. For this we need to show how the two agents can successfully make a turn at any point without deviating (in expectation) from that point. 

Indeed, consider a scheduled turning point and consider the two $p$-faulty agents approaching that point. For some $\gamma >0$ small enough, at time $2\gamma$ before the agents arrive at the point, the two agents undertake two distinguished roles, that of a Leader and that of a Follower. The roles can remain invariant throughout the execution of the algorithm. 
The Follower instantaneously slows down so that when the distance of the Leader and the Follower becomes $2\gamma$, the Follower is $\gamma/(1-p)$ before the turning point (which is strictly more than $\gamma$ and strictly less than $2\gamma$), and as a result the Follower has passed the turning point by $\frac{1-2p}{1-p}\gamma$. 
 At this moment, the Follower resumes full speed, and both agents move towards the same direction as before. Then, the Leader runs Algorithm~\ref{algo:subroutine change direction} with mobile Follower being the other $p$-faulty agent. 

Note that if at any moment the two agents meet, it is because after a successful turn the two have moved towards each other for time $\gamma$, whereas if a turn is unsuccessful the two preserve their relative distance $\gamma$. Therefore, since the moment of the first turning attempt, the two agents meet in expected time
$
\sum_{i=0}^\infty (i+1)\gamma (1-p)p^i =  \frac{\gamma}{1-p}.
$
We conclude that the expected meeting (and turning) point is the original turning point. 
Most importantly, the probability that the resulting turning point is away from the given turning point drops exponentially with $p$, and is also proportional to $\gamma$, which can be independently chosen to be arbitrarily small.
\qed \end{proof}

Lemma~\ref{lem: 2 agents simulate 1} is quite powerful, since it shows how to simulate deterministic turns  with arbitrarily small deviation from the actual turning points. More importantly, that deviation can be chosen to drop arbitrarily fast, dynamically, hence we can achieve smaller expected deviation later in the execution of the algorithm, compensating this way for the passed time. Therefore, we obtain the following theorem. 

\begin{theorem}
\label{thm: 2 agents 9 competitive}
For all $p \in [0,1/2)$, two deterministic faulty agents can solve \search{p} with competitive ratio $9+\epsilon$, for every $\epsilon>0$, independently of the underlying communication model. Also, the performance is concentrated arbitrarily close to $9+\epsilon$. 
\end{theorem}

Is it worthwhile noticing that agents' movements, in the underlying algorithm of Theorem~\ref{thm: 2 agents 9 competitive} is still probabilistic, due to the probabilistic faulty turns. However, choosing appropriate parameters every time Algorithm~\ref{algo:subroutine change direction} is invoked, one can achieve arbitrary concentration in the expected performance of the algorithm, hence the bound of $9+\epsilon$ can be practically treated as deterministic. 

In contrast, using the same trick, we can achieve a much better competitive ratio, but only in expectation equal to the one of~\cite{kao1996searching} (with uncontrolled concentration). To see how, note that by the proof of Theorem~\ref{thm: deterministic simulates randomized}, the two $p$-faulty agents can stay together in order to collect sufficiently many random bits and simulate any finite number of queries to a random oracle. 
Then using Lemma~\ref{lem: 2 agents simulate 1}, the agents simulate the optimal randomized algorithm of ~\cite{kao1996searching} with performance $4.59112$, designed originally for one randomized $0$-faulty agent that makes only 2 queries to the uniform distribution. In other words, the two deterministic $p$-faulty agents can overcome their faulty turns using Lemma~\ref{lem: 2 agents simulate 1} and the lack of random oracle by invoking Theorem~\ref{thm: deterministic simulates randomized}. To conclude, we have the following theorem which requires that $p>0$. 

\begin{theorem}
\label{thm: 2 agents 4.59112 competitive}
Two deterministic faulty agents can solve \search{p} with competitive ratio $4.59112+\epsilon$, for every $\epsilon>0$ and for every $p\in (0,1/2)$, independently of the underlying communication model. 
\end{theorem}

In our final main result we show that two $p$-faulty agents operating in the wireless model can do better than $9$ for all $p<1/2$, as well as better than $4.59112$ for a large spectrum of $p$ values. Note that the achieved competitive ratio is at least $3$, which is the optimal competitive ratio for searching with two $0$-faulty agents in the wireless model, and that our result matches this known bound when $p\rightarrow 0$.

\begin{theorem}
\label{thm: 2 agents improved competitive}
Two deterministic $p$-faulty agents in the wireless model can solve \search{p} with competitive ratio $3+4\sqrt{p(1-p)}+\epsilon$, for every $\epsilon>0$ and for every $p\in [0,1/2)$.
\end{theorem}

\vspace{-.5cm}

\begin{algorithm}
    \caption{Search with two wireless $p$-faulty agents} 
    \label{algo: search with two wireless}
    \begin{algorithmic}[1]
        \REQUIRE $p$-faulty agents with distinct roles of Leader and Follower, $s <1$ and $\gamma>0$.
        \STATE Agents search in opposite direction until target is found and reported. 
	 \STATE Target finder becomes Leader, and non-finder becomes Follower. 
        \STATE Non-finder changes speed to $s$, attempts a turn (that fails with probability), and continues moving until she meets with the finder. 
        \STATE Finder moves in same direction for $\gamma>0$ and runs Algorithm~\ref{algo:subroutine change direction} (target plays role of Follower), until the target is reached again. 
        \STATE Finder (Leader) continues until she meets the non-finder. 
        \STATE Non-finder (Follower) stays put until met by the Finder (Leader) again. 
        \STATE Leader continues moving in the same direction (away from the target and the Follower) for time $\gamma$ and then runs Algorithm~\ref{algo:subroutine change direction} with the Follower being the immobile agent, in order to turn. 
	\STATE When the Leader turns successfully, she picks up the Follower, and continuing  in the same direction, together, they move to the target. 
    \end{algorithmic}
\end{algorithm}
\vspace{-.5cm}

\begin{proof}[ (sketch) of Theorem~\ref{thm: 2 agents improved competitive}; See Section~\ref{sec: proof of thm: 2 agents improved competitive} for the details.]
The proof is given by the performance analysis of Algorithm~\ref{algo: search with two wireless} for a proper choice of speed $s<1$. In this simplified (sketch of) proof, we make the assumption that the target finder (using the target) as well as the two agents when walking together can make a deterministic turn. Indeed, using Algorithm~\ref{algo:subroutine change direction}, we show, in the full proof (Section~\ref{sec: proof of thm: 2 agents improved competitive}) how the actual probabilistically faulty turns have minimal impact in the competitive ratio.

We assume that the target is reported by the finder at time 1, when the distance of the two agents is 2. 
As for the non-finder, she turns successfully when she receives the wireless message with probability $1-p$. Since the finder moves towards her, their relative speed $1+s$. This means that they meet in additional time $2/(1+s)$, during which time the non-finder has moved closer to the target by $2s/(1+s)$. Hence, when the two agents meet, they are at distance $2-2s/(1+s)$ from the target. 

On the other hand with probability $p$ the non-finder fails to turn, and the two agents continue to move towards the same direction, only that the non-finder's new speed is $s$. So, their relative speed in this case is $1-s$. This means that they meet in additional time $2/(1-s)$, during which time the non-finder has moved further from the target by $2s/(1-s)$. Hence, when the two agents meet, they are at distance $2+sW+s(2+sW)/(1-s)$ from the target. Also recall that when they meet, they can make together a forced turn (that affects minimally the termination time), inducing this way total termination time (and competitive ratio, since the target was at distance 1)
$$
3+
p 
\left(
\frac{2}{1-s} +\frac{2s}{1-s}
\right)
+
(1-p) 
\left(
\frac{2}{1+s} -\frac{2s}{1+s}
\right)
=
\frac{5-s (s+4-8 p)}{1-s^2}
.$$
We choose $s=s(p) = \frac{1-2 \sqrt{p-p^2}}{1-2 p}$, the minimizer of the latter expression, that can be easily seen to attain values in $(0,1)$ for all $p \in (0,1/2)$, hence it is a valid choice for a speed.
Now we substitute back to the formula of the competitive ratio, and after we simplify algebraically, the expression becomes
$
3+4 \sqrt{(1-p) p}
$.
\qed \end{proof}

It is worthwhile noticing that the upper bound $4.59112$ of Theorem~\ref{thm: 2 agents 4.59112 competitive} holds in expectation, without being able to control the deviation. However, the upper bound of Theorem~\ref{thm: 2 agents improved competitive} holds again in expectation, but the resulting performance can be concentrated around the expectation with arbitrary precision. 
Moreover, the derived competitive ratio is strictly increasing in $p<1/2$, and ranges from 3 to 5. Hence, the drawback of Theorem~\ref{thm: 2 agents improved competitive} is that for high enough values of $p$ ($p> 0.197063$), the induced competitive ratio exceeds $4.59112$. 
Therefore, we have the incentive to choose either the algorithm of Theorem~\ref{thm: 2 agents improved competitive}, when
$p\leq 0.197063$, and the algorithm of Theorem~\ref{thm: 2 agents 4.59112 competitive} otherwise. 
It would be interesting to investigate whether a hybrid algorithm, combining the two ideas, could accomplish an improved result

\section{Conclusion}

In this paper we studied a new mobile agent search problem whereby an agent's ability to navigate in the search space exhibits probabilistic faults in that every attempt by the agent to change direction is an independent Bernoulli trial (the agent fails to turn with probability $p<1/2$). 
When searching with one agent, our best performing algorithm has optimal performance $4.59112$ as $p\rightarrow 0$, performance less that $9$ for $p\leq 0.436185$, and optimal performance up to constant factor and unbounded as $p\rightarrow 1/2$. 
When searching with two faulty agents, we provide 3 algorithms with different attributes. One algorithm has (expected) performance $9$ with arbitrary concentration, the other has performance $4.59112$, and finally one has performance $3+4\sqrt{p(1-p)}$ (ranging between 3 and 5) again with arbitrary concentration. 

It is rather surprising that even in this probabilistic setting with one searcher, we can design algorithms that outperform the well-known zig-zag algorithm for linear search whose competitive ratio is $9$, as well as that the problem with two searchers admits bounded competitive ratio for all $p\in (0,1/2)$, and unlike the one search problem. 
Interesting questions for further research could definitely arise in the study of similar, related ``probabilistic navigation'' faults either for their own sake or in conjunction with ``communication'' faults in more general search domains (e.g., in the plane or more general cow path with $w$ rays) and for multiple (possibly collaborating) agents.

\bibliographystyle{abbrv}

\bibliography{refs,refs1}

\newpage

\appendix

\section{Proof of Lemma~\ref{lem: unbounded cr}}
\label{sec: proof of lem: unbounded cr}

\begin{proof}[sketch of Lemma~\ref{lem: unbounded cr}]
Consider a trajectory and some moment at which the agent makes a turn while moving away from the origin, while at distance $d>0$. 
We calculate a lower bound for the time it takes the agent to return to the origin, which is also a lower bound for the time to expand the searched space in the opposite direction. 

Note that with probability $1-p$ the turn is successful, and the agent could return to the origin in time at least $d$. Also with probability $p$ the turn has failed, and the agent cannot know of the deviation to the intended trajectory earlier than time $d$ (because it needs time at least $d$ to reach the origin). In the latter event, the agent must reach distance $2d$ from the origin. Similarly, when the next turn is attempted, we have that with probability $1-p$ the agent needs at least time $2d$ to reach the origin, and will not know of the possible failed turn (with prob $p$) unless it reaches distance $4d$ from the origin, and so on. Therefore, the expected time for the agent to reach the origin is at least 
$
\sum_{i\geq0} (1-p)p^i 2^i d. 
$
The series is divergent when $p\geq 1/2$, concluding the proof of our claim. 
\qed \end{proof}

\section{Proof of Theorem~\ref{thm: deterministic simulates randomized}}
\label{sec: proof of thm: deterministic simulates randomized}

In order to prove Theorem~\ref{thm: deterministic simulates randomized} we will need the following lemma that shows how to amplify the probability of faults. 

\begin{lemma}
\label{lem: probability amplification}
For any $p\in (0,1)$, and for any $\epsilon>0$, there exists $\delta\in (0,\epsilon)$ such that a deterministic $p$-faulty agent can simulate a deterministic $q$-faulty agent, where $|q-1/2|\leq \delta$. 
\end{lemma}

\begin{proof}
Consider a $p$-faulty agent, where $p \in (0,1)$. Suppose that its current direction is to the right, and it attempts to change direction. The probability that the direction remains the same is $p$. Call this $q_1$, where the index refers to 1 attempt to turn. 

Now consider $k\geq 2$ consecutive Bernoulli trials, in which the agent attempts to change direction $k$ consecutive times without moving in between the attempts. We compute the probability $q_k$ that the direction is still to the right. For this event to happen, either in the $k-1$ attempt the direction is to the left (which happens with probability $1-q_{k-1}$) and in the last trial the direction changes as attempted (with probability $1-p$), or in the $k-1$st attempt the direction is to the right (which happens with probability $q_{k-1}$) and in the last trial the direction does not change as attempted (with probability $p$). Hence we have the recurrence relation 
$$
q_k=(1-p)(1-q_{k-1})+p~q_{k-1}
$$
with initial condition $q_1=p$. The solution to the recurrence is 
$$
q_k=\frac{1}{2}\left((2 p-1)^k+1\right) .
$$
Since $p\in (0,1)$, we have that for any $\epsilon>0$, there exists a big enough $k\in \naturals$, such that $1/2-\epsilon<q_k<1/2+\epsilon$. 
\qed \end{proof}

We are not ready to prove Theorem~\ref{thm: deterministic simulates randomized}. 

\begin{proof}[of Theorem~\ref{thm: deterministic simulates randomized}]
By Lemma~\ref{lem: probability amplification} we may assume that we have a deterministic $q$-faulty agent, where $q$ is arbitrarily close to 1/2. Hence it is enough to show how a deterministic $q$-faulty agent can simulate the random choices of a randomized $p$-faulty agent (algorithm). 

Recall that the randomness of the $p$-faulty agent (algorithm) is induced by the access to an oracle sampling from the uniform distribution. Therefore, it is enough to show how to use the random faults of (an otherwise) deterministic $q$-faulty agent in order to sample from the uniform distribution. 
For this it is enough to show how to generate (nearly) uniform binary strings of arbitrary length. 

Indeed, fix any integer $k$ (that can be arbitrarily large; aiming for a string of length $k$). For any $\gamma>0$, we show how the deterministic $q$-faulty agent can generate a nearly uniform binary string of length $k$ without leaving further than $\gamma$ away from the origin.\footnote{Since we defined the competitive ratio resulting from placements of the exit arbitrarily away from the origin, $\gamma$ does not need to be small, but we keep it small for the sake of intution.}
For this the agent moves in any direction for time $\zeta>0$, for some $\zeta<\gamma/2^k$ small enough. Then it attempts a turn and moves for additional time $\zeta$ to either arrive back to the origin (with probability $1-q$) or not (with probability $q$). This way the agent has collected one random bit which is 1 with probability $q$. If the agent is at the origin, it repeats the process to collect more random bits. If not, then it is provably $2\zeta$ away from the origin. It then attempts to turn again, and moves for a total additional time $2\zeta$ in which case either it reaches the origin (with probability $1-q$) or not (with probability $q$), collecting another random bit. Note that it is important here that the probabilistic faults are independent Bernoulli trials. 
The agent continues the same way until it collects $k$ random bits, having moved at most $2^k \zeta<\gamma$ away from the origin. 

To summarize, we showed above how a deterministic $p$-faulty agent can (in an ad-hoc way) act as a deterministic $q$-faulty agent ($q$ being arbitrarily close to 1/2), and hence collecting any arbitrarily large nearly uniform binary string, and most importantly staying arbitrarily close to the origin. Using the random binary string, the agent can simulate any finite number of queries to the uniform distribution (with arbitrary precision). 
Hence it remains to have the deterministic $p$-faulty agent return to the origin in order to execute the trajectory (i.e. to simulate the algorithmic random choices) of the randomized $p$-faulty agent. From the previous step, the deterministic agent knows how far away from the origin is (if not already at the origin). Call this distance $w\leq \gamma$. Recall that $\gamma$ could have been chosen to be arbitrarily small, and it's value may also depend on $p$. 
The agent can make repetitive attempts to return to the origin, each time (if failing) doubling its distance to the origin. Indeed, after $l$ attempts, the probability that the agent has not arrived back to the origin is $p^l$, in which case it is at distance at most $2^l w<2^l \gamma$ from the origin. In that event, the $p$-faulty agent can run the deterministic algorithm of Theorem~\ref{thm: det algo limsup} whose competitive ratio we denote by $c^\textsc{det}$, inducing competitive ratio arbitrarily close to $c^\textsc{det}$ (because the initial point of the agent is arbitrarily close to the origin). In the complementary event, the agent has returned to the origin, with probability $1-p^l$. Hence the competitive ratio is
$
(1-q^l) c + q^l ~ c^\textsc{det}
$
which is arbitrarily close to $c$ as wanted. 
\qed \end{proof}

\section{Proof of Lemma~\ref{lem: cr deterministic in interval}}
\label{sec: proof of lem: cr deterministic in interval}

In this section we prove Lemma~\ref{lem: cr deterministic in interval}.
We emphasize that references to the direction of movement below (and in all sections) is something relative. Since the agent can distinguish the origin we use the direction of movement as going towards the left or right and we assume without loss of generality that on the right we have the positive points and to the left we have the negative points. However, the agent is actually agnostic of this fact since it just counts the time that has passed to compute its next move.

For any placement $n$ of the target, there is $t \in \naturals$ such that $n \in (g^t, g^{t+1}]$. For any initial intended expansion, we let $T$ be the random variable that is equal to the termination time of the algorithm for the target placement $n$.  
When the agent passes from the origin, we denote by $X$ the random variable that equals the intended expansion, so that the expected termination time of the given algorithm is $E[T|X=g^0]$.
We set 
\begin{align*}
R_k &:=E[T|X=g^k] \\
L_k &:=E[T|X=-g^k],
\end{align*}
that is $R_k, L_k$ are the expected termination time of the algorithm condition on that the agent will initially expand up to $g^k$ (moving to the right) and $-g^k$ (moving the left), respectively. Therefore we want to compute $R_0=E[T|X=g^0]$. 

Now suppose that the agent is at the origin and intends to expands up to $g^k$, with $k\leq t$. We have the following 
\begin{align*}
R_k &=
(1-p)
\left( 2g^k + L_{k+1}  \right)
+
(1-p)p
\left( 2g^{k+1} + L_{k+2}  \right)
+
\ldots \\
& ~~~+
(1-p) p^{t-k}
\left( 2g^t + L_{t+1}  \right) 
+
(1-p)\sum_{i=t-k+1}^\infty p^i n \\
&=
(1-p)
\sum_{i=0}^{t-k-1}p^i L_{k+i+1}
+(1-p)p^{t-k}L_{t+1}\\
& ~~~+(1-p)2g^k\sum_{i=0}^{t-k}(pg)^i
+
(1-p)\sum_{i=t-k+1}^\infty p^i n. 
\end{align*}

Note also that for all $k\geq t+1$, we have that $R_k=n$. 
With this observation, we have that for all $k\leq t+1$,

\begin{align*}
L_k &=
(1-p)
\left( 2g^k + R_{k+1}  \right)
+
(1-p)p
\left( 2g^{k+1} + R_{k+2}  \right)
+
(1-p)p^2
\left( 2g^{k+2} + R_{k+3}  \right)
+
\ldots \\
&= (1-p)2g^k 
\sum_{i=0}^\infty (pg)^i 
+
(1-p)
\sum_{i=0}^\infty p^i R_{k+1+i}  \\
&= (1-p)2g^k 
\sum_{i=0}^\infty (pg)^i 
+
(1-p)
\sum_{i=0}^{t-k-1} p^i R_{k+1+i}  
+
(1-p)
\sum_{i=t-k}^\infty p^i n
\end{align*}
which in particular implies that 
$$
L_{t+1}
=
(1-p)2g^{t+1} 
\sum_{i=0}^\infty (pg)^i 
+
(1-p)
\sum_{i=0}^\infty p^i n
=
\frac{2 (1-p) g^{t+1}}{1-g p}+n.
$$

Next we set, for $k=0,\ldots,t$, 
\begin{align}
\alpha_k := &
(1-p)p^{t-k}L_{t+1}
+(1-p)2g^k\sum_{i=0}^{t-k}(pg)^i
+
(1-p)\sum_{i=t-k+1}^\infty p^i n \notag \\
= &
(1-p)p^{t-k}\left( \frac{2 (1-p) g^{t+1}}{1-g p}+n \right)
+(1-p)2g^k\left( \frac{1-(g p)^{t-k+1}}{1-g p} \right) 
+
n p^{t-k+1} \notag 
\\
= &
p^{t-k} \frac{2 (1-p)^2 g^{t+1}}{1-g p}
+(1-p)2g^k\left( \frac{1-(g p)^{t-k+1}}{1-g p} \right) 
+
n p^{t-k} \notag\\
= &
\frac{2 (1-p)g^k }{1-g p}\left(1 + (1-2 p)g (pg)^{t-k}\right)
+
n p^{t-k} \label{equa: alphak} \\
\beta_k := &
(1-p)2g^k 
\sum_{i=0}^\infty (pg)^i 
+
(1-p)
\sum_{i=t-k}^\infty p^i n \notag \\
= & 
\frac{2 (1-p) g^k}{1-g p}+n p^{t-k}, \label{equa: betak}
\end{align}
as well as we also set $\alpha^T=(\alpha_0, \ldots, \alpha_t), \beta^T=(\beta_0, \ldots, \beta_t)$. 

With that notation in mind, we have that $R_0, \ldots, R_t, L_0, \ldots, L_{t}$ satisfy the following linear system of $2(t+1)$ constraints
\begin{align}
R_k - (1-p)
\sum_{i=0}^{t-k-1}p^i L_{k+i+1}
=& \alpha_k, ~~k=0,\ldots,t, \label{equa: R for L}\\
L_k - (1-p)
\sum_{i=0}^{t-k-1} p^i R_{k+1+i}  
=& \beta_k, ~~k=0,\ldots,t. \label{equa: L for R}
\end{align}
Next we define matrix $P \in \reals^{(t+1)\times(t+1)}$ with $P_{i,j}=-(1-p)p^{j-i-1}$ if $j>i$ and $0$ otherwise, that is
$$
P=
-(1-p) \left(
\begin{array}{ccccccc}
0 & p^0 & p^1 & p^2 & \ldots & p^{t-2} & p^{t-1}  \\
0 & 0   & p^0 & p^1  & \ldots & p^{t-3} & p^{t-2}  \\
\vdots & \vdots & \vdots & \vdots & \ddots & \vdots & \vdots \\
0 & 0 & 0 & 0 & \ldots & p^0 & p^1 \\
0 & 0 & 0 & 0 & \ldots & 0 & p^0 \\
0 & 0 & 0 & 0 & \ldots & 0 & 0
\end{array}
\right),
$$
so that the linear system~\eqref{equa: R for L},~\eqref{equa: L for R} can be written in the matrix form $Ax=b$, where 
\begin{equation}
\label{equa: matrix A}
A:=
\left(
\begin{array}{cc}
I_{t+1} & P \\
P & I_{t+1}
\end{array}
\right),
\end{equation}
$I_{t+1}$ denotes the $(t+1)\times(t+1)$ identity matrix, $x^T=(R_0, R_1, \ldots, R_t, L_0, L_1, \ldots, L_t)$, and $b^T=(\alpha^T,\beta^T)$. 
For simplicity, we abbreviate $I_{t+1}$ by $I$. 

\begin{lemma}
\label{lem: inverse of block}
If matrix $I-P^2$ is invertible, then block matrix $A$ is invertible and  
$$
A^{-1}=
\left(
\begin{array}{cc}
(I-P^2)^{-1} & - (I-P^2)^{-1}P \\
P(I-P^2)^{-1} & I + P(I-P^2)^{-1}P
\end{array}
\right)
$$
\end{lemma}

\begin{proof}
Straightforward matrix multiplication. 
\qed \end{proof}
Our goal next is to compute $A^{-1}$. We find the inverse in block form. 
\begin{lemma}
\label{lem: inverse of I-P2}
Let $\delta_1=1$ and 
$$
\delta_m=\frac{1-p}2 \left( 
1 - (-1 + 2 p)^{m-2}
\right),
~\textrm{for all}~m\geq 2.
$$
Then matrix $B:=I-P^2$ is invertible, and 
$$
B_{i,j}^{-1}
=
\left\{
\begin{array}{ccccc}
0 &~\textrm{if}~ i>j \\
\delta_{j-i+1} &~\textrm{if}~ j\geq i \\
\end{array}
\right.
$$
\end{lemma}

\begin{proof}[Sketch]
It is easy to see that $\left(I-P^2\right)_{i,i}=1$, and that $\left(I-P^2\right)_{i,i+1}=0$. For all $j\geq i+2$ we have that 
\begin{align*}
\left(
I-P^2
\right)_{i,j}
&=
-\left(
P^2
\right)_{i,j}
=
-\sum_{r=1}^{t+1}P_{i,r} P_{r,j}
=
-\sum_{r=i+1}^{j-1}P_{i,r} P_{r,j} \\
& =
-(1-p)^2 \sum_{r=i+1}^{j-1} p^{r-i-1} p^{j-r-1}
=
-(j-i-1) (1-p)^2 p^{j-i-2}.
\end{align*}

Next we show that 
$\left(I-P^2\right)B=B\left(I-P^2\right)=I$. 
The proof is by induction on $t$ (recall that $B \in \reals^{(t+1)\times(t+1)}$).

The case of $t=1$ is straightforward since then we have $B=I-P^2=I$. 
For the case $t\geq 2$ (when $B \in \reals^{(t+1)\times(t+1)}$), we note that 
the principal minor of $B$ is the same as the matrix induced by taking $t\leftarrow t-1$. 

Hence, in order to show that $B \left(I-P^2\right)=I$, it is enough to prove, for all $r=0, \ldots,t$, that
$$
\left(  B \left(I-P^2\right)  \right)_{r,t+1}
=
\left(  B \left(I-P^2\right)  \right)_{t+1,r}
=
\left( \left(I-P^2\right)  B \right)_{r,t+1}
=
\left(  \left(I-P^2\right) B \right)_{t+1,r}
=0, 
$$
as well as that 
$$
\left(  B \left(I-P^2\right)  \right)_{t+1,t+1}
=
\left(  \left(I-P^2\right)B  \right)_{t+1,t+1} = 1.
$$

We show the first equality when $r=1$, and the rest admit very similar calculations. Indeed, for $r=1$ we have
\begin{align*}
\left(
B \left(I-P^2\right)
\right)_{1,t+1}
=&
\sum_{r=1}^{t+1}B_{1,r} (I-P^2)_{r,t+1} \\
= &
-\sum_{r=1}^{t-1}B_{1,r} (t-r) (1-p)^2 p^{t-r-1} + B_{1,t+1}\\
= &
- (t-1) (1-p)^2 p^{t-1-1}\\
& -
\sum_{r=3}^{t-1}\frac{1-p}2 \left( 
1 - (-1 + 2 p)^{r-2}
\right)
 (t-r) (1-p)^2 p^{t-r-1}\\
& +
 \frac{1-p}2 \left( 
1 - (-1 + 2 p)^{t-1}
\right) \\
= & ~0.
\end{align*}
~
\qed 
\end{proof}

The following is an immediate corollary of Lemma~\ref{lem: inverse of I-P2}.

\begin{corollary}
\label{cor: inverse of I-P2 times B}
$\left( (I-P^2)^{-1}P \right)_{1,1}=0$ and for $r=2, \ldots, t+1$ we have $\left( (I-P^2)^{-1}P \right)_{1,r} =-\frac{1}{2} (1-p) \left((2 p-1)^{r-2}+1\right)$.
\end{corollary}

\begin{proof}
Using Lemma~\ref{lem: inverse of I-P2} we calculate 
\begin{align*}
\left( (I-P^2)^{-1}P \right)_{1,r}
=&
\sum_{i=1}^{t+1}
\left( (I-P^2)^{-1}\right)_{1,i}
P_{i,r} \\
= &
P_{1,r}
-\frac12(1-p)^2
\sum_{i=3}^{r-1}
\left( 
1 - (-1 + 2 p)^{i-2}
\right)
p^{r-i-1}
\end{align*}
It follows that 
\begin{align*}
\left( (I-P^2)^{-1}P \right)_{1,1} =& P_{1,1}=0 \\
\left( (I-P^2)^{-1}P \right)_{1,2} =& P_{1,2}=-(1-p) \\
\left( (I-P^2)^{-1}P \right)_{1,3} =& P_{1,3}=-p(1-p),
\end{align*}
and for all $i\geq 4$, 
\begin{align*}
\left( (I-P^2)^{-1}P \right)_{1,i}
&=
P_{1,r}
-\frac12(1-p)^2
\sum_{i=3}^{r-1}
\left( 
1 - (-1 + 2 p)^{i-2}
\right)
p^{r-i-1} \\
&=
-\frac{1}{2} (1-p) \left((2 p-1)^{r-2}+1\right).
\end{align*}
\qed \end{proof}

We are ready to calculate $R_0$. By~\eqref{equa: matrix A} and Lemma~\ref{lem: inverse of block}, we have that 
\begin{align*}
R_0  & = 
 A^{-1}  
\left(
\begin{array}{c}
\alpha \\
\beta
\end{array}
\right) \\
& =
\sum_{j=1}^{t+1} \left( (I-P^2)^{-1} \right)_{1,j} \alpha_{j-1} 
-
\sum_{j=1}^{t+1} \left( (I-P^2)^{-1} P \right)_{1,j} \beta_{j-1}  \\
& =
\alpha_0
+
\sum_{j=2}^{t+1} \left( (I-P^2)^{-1} \right)_{1,j} \alpha_{j-1} 
-
\sum_{j=2}^{t+1} \left( (I-P^2)^{-1} P \right)_{1,j} \beta_{j-1}. 
\end{align*}
Now, using Lemma~\ref{lem: inverse of I-P2}, and Corollary~\ref{cor: inverse of I-P2 times B}, we continue the calculations of the latest expression as
\begin{align*}
R_0  & = 
\alpha_0
+
\frac{1-p}2
\left(
\sum_{j=2}^{t+1} 
\left( 
1 - (2 p-1)^{j-2}
\right)
 \alpha_{j-1} 
+
\sum_{j=2}^{t+1} 
\left(1+(2 p-1)^{j-2}\right)
\beta_{j-1}. 
\right) \\
& = 
\alpha_0
+
\frac{1-p}2
\left(
\sum_{j=2}^{t+1} 
\left(\alpha_{j-1}+\beta_{j-1}\right)
+
\sum_{j=2}^{t+1} 
(2 p-1)^{j-2}
\left(
\beta_{j-1}-\alpha_{j-1}
\right)
\right)
\end{align*}
Using now~\eqref{equa: alphak} and~\eqref{equa: betak} and expanding on the formula before, we have that 
\begin{equation} \label{eq:r zero}
    R_0 =
    \frac
    {
    (1 - p) g^{t+1} \left(    
    1+
    (1-2 p) \left(g+(g-1) (-1+2 p)^t \right)
    \right)
    }
    {(g-1) (1-g p)}
    - 
    \frac{2 (1-p)  }
    {g-1}
    +
    n.    
\end{equation}
We observe that for all $t$ and $g\geq 2$, we have $g+(g-1) (-1+2 p)^t \geq 0$. Therefore,
\begin{align*}
\sup_{n \in (g^t, g^{t+1}]} \frac{R_0}{n} 
& = \frac{R_0}{g^t} \\
& \leq 
\frac{
(1 - p) g 
}
{(g-1) (1-g p)}
\left(    
1+
(1-2 p) \left(g+(g-1) (-1+2 p)^t \right)
\right)
+
1
\\
\ignore{
&=
\frac{
(1 - p) g 
\left(    
1+
(1-2 p) g
\right)
}
{(g-1) (1-g p)}
+
1
+
\frac{
(1 - p) (1-2p) g 
}
{(1-g p)}
\left(-1+2 p\right)^t
\\
}
&=
\frac{
(1 - p) g 
}
{(1-g p)}
\left(
\frac{
1+
(1-2 p) g
}
{
g-1
}
-
\left(-1+2 p\right)^{t+1}
\right)
+1
\\
&\leq
\frac{
(1 - p) g 
}
{(1-g p)}
\left(
\frac{
1+
(1-2 p) g
}
{
g-1
}
+
\left(1-2 p\right)^{t+1}
\right)
+1
\end{align*}
as promised. 

\section{Proof of Lemma~\ref{lem: cr randomized in interval}}
\label{sec: proof of lem: cr randomized in interval}


As in the proof of Lemma~\ref{lem: cr deterministic in interval} we introduce the following notation (and right, e.g., $R(x)$, instead of $R_x$ so as to distinguish the proofs). 
\begin{definition}
    Let $R(x)$ be the expected termination time when moving to the right and the initial intended turning point is $g^{x+\epsilon}$, where $x$ is a non-negative integer.
\end{definition}

\begin{definition}
    Let $L(x)$ be the expected termination time when moving to the left and the initial intended turning point is $g^{x+\epsilon}$, where $x$ is a non-negative integer.
\end{definition}

Given the above definitions and that initially the agent starts moving to the left or right with probability $1/2$ we can compute the total expected termination time as 
$$\frac{1}{2}L(0) + \frac{1}{2}R(0).$$
Then, we will be able to find the expected competitive ratio by dividing with the distance of the (worst placed) target (arbitrarily away) from the origin.

We now assume that the target is placed to the right of the origin\footnote{If the target is to the left, then the analysis is analogous.} at $n = g^{t+\delta}$, where $t$ is a non-negative integer and $\delta \in (0,1]$. This gives us two possible cases of when the agent will reach the target. One is when $\epsilon < \delta$ and the other when $\epsilon \geq \delta$, where $\epsilon$ is random algorithmic choice sampled from $[0,1]$ uniformly at random. 

As a result, we is natural to associate $L(x)$ with random variables $L_{\epsilon < \delta}$ and $L_{\epsilon \geq \delta}$, as well as $R(x)$ with $R_{\epsilon < \delta}$ and $L_{\epsilon \geq \delta}$, corresponding to the two cases $\epsilon < \delta$ (see Section~\ref{sec: epsilon < delta}) and $\epsilon \geq \delta$  (see Section~\ref{sec: epsilon >= delta}), respectively. In this case we have 
\begin{align*}
L(0) &= \int_{0}^{\delta} L_{\epsilon < \delta}(0) + \int_{\delta}^{1} L_{\epsilon \geq \delta}(0), \\
R(0) &= \int_{0}^{\delta} R_{\epsilon < \delta}(0) + \int_{\delta}^{1} R_{\epsilon \geq \delta}(0).
\end{align*}

\subsection{First Case ($\epsilon < \delta$)}
\label{sec: epsilon < delta}

Functions $L,R$ in this section are assumed to refer specifically to the case that $\epsilon < \delta$, i.e. their are abbreviations of $L_{\epsilon < \delta}$  and $R_{\epsilon < \delta}$, respectively. 
Next we represent $R(x), L(x)$ recursively. Indeed, we have 

\begin{align*}
    R(t+1) &= n, ~~~~\textit{(Since intended turning point is } g^{t+1+\epsilon} > g^{t+\delta} = n\textrm{)} \\
    R(t)   &= pR(t+1) + (1-p)(2g^{t+\epsilon} + L(t+1)) \\
           &= pn + (1-p)(2g^{t+\epsilon} + L(t+1)) \\
    R(t-1) &= pR(t) + (1-p)(2g^{t-1+\epsilon} + L(t)) \\
           &= p^2n + (1-p)p(2g^{t+\epsilon} + L(t+1)) + (1-p)(2g^{t-1+\epsilon} + L(t))
\end{align*}
Moreover, for any $x \geq t+1$ we have that $R(x) = n$ because the intended turning point is located at a greater distance than the target, hence the target is found before reaching the intended turning point.

For $x=k$ for some $k<t+1$ we distinguish two cases. One case is that the change of direction succeeds (this happens with probability $1-p$). In that case the agent moves up to $g^{k+\epsilon}$ and back to the origin and from there we need to add the expected termination time of $L(k+1)$. The other case is that the change of direction fails (this happens with probability $p$). Then, the agent moves to the next intended turning point (i.e., $g^{k+1+\epsilon}$). This case is identical to the scenario of the agent moving to $g^{k+1+\epsilon}$ which is actually $R(k+1)$. Therefore, the expected termination time moving to the right, in the general case of $x=k$ for some $k<t+1$, is
$$R(k) = pR(k+1) + (1-p)(2g^{k+\epsilon} + L(k+1)).$$
By expanding this recursive expression we get the following recursive expression that depends only on $L(x)$.
\begin{equation} \label{eq: expected time right}
    R(k) = np^{t+1-k} + (1-p)2g^{k+\epsilon}\sum_{i=0}^{t-k}(pg)^i + (1-p)\sum_{i=0}^{t-k}p^iL(k+1+i)
\end{equation}

Following a similar argument, we compute $L(x)$. The difference is that when the agent moves to the left there is no target, hence the agent will move until it actually changes direction. When the agent changes direction it reaches the origin and from there we add the expected time for moving to the right which is $R(x)$ for some $x$ based on the distance covered from the left.

\begin{align*}
    L(t+1) &= (1-p)(2g^{t+1+\epsilon} + R(t+2)) + (1-p)p(2g^{t+2+\epsilon} + R(t+3)) + \ldots \\ 
         &= (1-p)2g^{t+1+\epsilon}\sum_{i=0}^{\infty}(pg)^i + (1-p)\sum_{i=0}^{\infty}p^i R(t+2+i) \\
         &= \frac{(1-p)2g^{t+1+\epsilon}}{1-pg} + (1-p)\sum_{i=0}^{\infty}p^i n \\
         &= \frac{(1-p)2g^{t+1+\epsilon}}{1-pg} + n
\end{align*}
Similarly we have
\begin{align*}
    L(t)   &= (1-p)(2g^{t+\epsilon} + R(t+1)) + (1-p)p(2g^{t+1+\epsilon} + R(t+2)) + \ldots \\
           &= (1-p)2g^{t+\epsilon}\sum_{i=0}^{\infty}(pg)^i + (1-p)\sum_{i=0}^{\infty}p^i R(t+1+i) \\
    L(t-1) &= (1-p)(2g^{t-1+\epsilon} + R(t)) + (1-p)p(2g^{t+\epsilon} + R(t+1)) + \ldots \\ 
           &= (1-p)2g^{t-1+\epsilon}\sum_{i=0}^{\infty}(pg)^i + (1-p)\sum_{i=0}^{\infty}p^i R(t+i).
\end{align*}
From the above we  derive a general expression of $L(k)$ for some $k<t+1$ as follows
\begin{align*}
    L(k) &= (1-p)(2g^{k+\epsilon} + R(k+1)) + (1-p)p(2g^{k+1+\epsilon} + R(k+2)) + \ldots \text{ up to infinity} \\
           &= (1-p)\sum_{i=0}^{t-1-k}(2g^{k+\epsilon+i} + R(k+1+i)) + (1-p)\sum_{i=t-k}^{\infty}p^i (2g^{k+\epsilon+i} + R(k+1+i)) \\
           &= (1-p)2g^{k+\epsilon}\sum_{i=0}^{\infty}(pg)^i + (1-p)\sum_{i=0}^{t-k-1}R(k+1+i) + (1-p)\sum_{i=t-k}^{\infty}p^i n.
\end{align*}
Hence we have
\begin{equation}\label{eq: expected time left}
    L(k) = (1-p)2g^{k+\epsilon}\sum_{i=0}^{\infty}(pg)^i + (1-p)\sum_{i=0}^{t-k-1}R(k+1+i) + (1-p)\sum_{i=t-k}^{\infty}p^i n.
\end{equation}
Next we set, for $k=0,\ldots,t$,
\begin{align}
    \alpha_k :=& np^{t+1-k} + (1-p)2g^{k+\epsilon}\sum_{i=0}^{t-k}(pg)^i + (1-p)p^{t-k}L(t+1) \notag\\
              =& np^{t+1-k} - np^{t+1-k} + np^{t-k} + \frac{2g^{t+1+\epsilon}}{1-pg}(1-p)^2p^{t-k} 
    + \frac{2(1-p)g^{k+\epsilon}(1-(pg)^{t-k+1})}{1-pg} \notag\\
              =& \frac{2g^{t+1+\epsilon}(1-p)^2p^{t-k}}{1-pg} + \frac{2g^{k+\epsilon}(1-p)}{1-pg}
    - \frac{2g^{k+\epsilon}(1-p)(pg)^{t-k+1}}{1-pg} + np^{t-k} \notag\\
              =& \frac{2(1-p)g^{k+\epsilon}}{1-pg}\left((1-p)p^{t-k}g^{t+1-k}+1-(pg)^{t-k+1}\right) + np^{t-k} \notag\\
              =& \frac{2(1-p)g^{k+\epsilon}}{1-pg}\left(1+g(pg)^{t-k}-2(pg)^{t-k+1}\right) + np^{t-k} \notag\\
              =& \frac{2(1-p)g^{k+\epsilon}}{1-pg}\left(1+(1-2p)g(pg)^{t-k}\right) + np^{t-k}\label{eq: alpha plus epsilon} \\
    \beta_k :=& (1-p)2g^{k+\epsilon} \sum_{i=0}^\infty (pg)^i 
    + (1-p) \sum_{i=t-k}^\infty p^i n \notag \\
             =& \frac{2 (1-p) g^{k+\epsilon}}{1-pg}+n p^{t-k}, \label{eq: beta plus epsilon}
\end{align}
as well as we also set $\alpha^T=(\alpha_0, \ldots, \alpha_t), \beta^T=(\beta_0, \ldots, \beta_t)$. 

With that notation in mind and from equations \eqref{eq: expected time right}, \eqref{eq: expected time left}, \eqref{eq: alpha plus epsilon}, and \eqref{eq: beta plus epsilon}, we have that $R(0), \ldots, R(t)$ and $L(0), \ldots, L(t)$ satisfy the following linear system of $2(t+1)$ constraints
\begin{align}
R(k) - (1-p)
\sum_{i=0}^{t-k-1}p^i L(k+i+1)
=& \alpha_k, ~~k=0,\ldots,t, \label{eq: right for left}\\
L(k) - (1-p)
\sum_{i=0}^{t-k-1} p^i R(k+1+i)  
=& \beta_k, ~~k=0,\ldots,t. \label{eq: left for right}
\end{align}
The system has the same variable coefficient as the one considered on page~\pageref{equa: R for L} in the proof of Lemma~\ref{lem: cr deterministic in interval} in Section~\ref{sec: proof of lem: cr deterministic in interval}.
Since we solved the previous system by inverting the coefficient matrix, we use the same inverse, and elementary algebraic calculations (a matrix multiplication) to find 
\begin{equation} \label{eq:r zero smaller epsilon}
    R(0) = n + \frac{g^\epsilon (-1 + p) (-2 + 
    g (2 p + 
       g^t (1 + g - 2 g p - (-1 + g) (-1 + 2 p)^{1 + t})))}{(-1 + 
    g) (-1 + g p)}
\end{equation}
If we set $\epsilon=0$ in \eqref{eq:r zero smaller epsilon} we actually get the same result as in \eqref{eq:r zero}.
Since this is the case of $\epsilon < \delta$ in order to compute the expected termination time given that $\epsilon < \delta$ we have to compute
$$\int_{0}^{\delta} R_{\epsilon < \delta}(0) \, d\epsilon$$
which gives us
\begin{equation} \label{eq:integral of r zero smaller epsilon}
    n \delta + \frac{(-1 + g^\delta) (-1 + p) (-2 + 
    g (2 p + 
       g^t (1 + g - 2 g p - (-1 + g) (-1 + 2 p)^{(1 + t)})))}{(-1 + 
    g) (-1 + g p) \ln g}.
\end{equation}

\subsection{Second case ($\epsilon \geq 0$)}
\label{sec: epsilon >= delta}
Next we analyze the case where $\epsilon \geq \delta$. 
Again for notational convenience, only in this section we abbreviate 
$L_{\epsilon \geq \delta}$  and $R_{\epsilon \geq \delta}$ by $L,R$, respectively.
We compute $L(x), R(x)$ by establishing recurrence relations as before. 

\begin{align*}
    R(t) &= n, 
     ~~~~\textit{(Since intended turning point is } g^{t+\epsilon} > g^{t+\delta} = n \textrm{)} \\
        R(t-1) &= pR(t) + (1-p)(2g^{t-1+\epsilon} + L(t)) \\
           &= pn + (1-p)(2g^{t-1+\epsilon} + L(t)) \\
    R(t-2) &= pR(t-1) + (1-p)(2g^{t-2+\epsilon} + L(t-1)) \\
           &= p^2n + (1-p)p(2g^{t-1+\epsilon} + L(t)) + (1-p)(2g^{t-2+\epsilon} + L(t-1))
\end{align*}
Similarly, for any $x \geq t$ we have that $R(x)=n$.

Note that the current analysis is nearly identical to the previous case; indeed, the difference is that now the target is reached one step earlier. Therefore, the expected termination time moving to the right in the general case of $x=k$, for some $k<t$, is
$$R(k) = pR(k+1) + (1-p)(2g^{k+\epsilon} + L(k+1)).$$
By expanding this recursive expression we get the following recursive expression that depends only on $L(x)$.
\begin{equation} \label{eq: expected time right second case}
    R(k) = np^{t-k} + (1-p)2g^{k+\epsilon}\sum_{i=0}^{t-1-k}(pg)^i + (1-p)\sum_{i=0}^{t-1-k}p^iL(k+1+i)
\end{equation}

Similarly we compute $L(x)$ for some useful values of $x$. We have
\begin{align*}
    L(t) &= (1-p)(2g^{t+\epsilon} + R(t+1)) + (1-p)p(2g^{t+1+\epsilon} + R(t+2)) + \ldots \\
         &= (1-p)2g^{t+\epsilon}\sum_{i=0}^{\infty}(pg)^i + (1-p)\sum_{i=0}^{\infty}p^i R(t+1+i) \\
         &= \frac{(1-p)2g^{t+\epsilon}}{1-pg} + (1-p)\sum_{i=0}^{\infty}p^i n \\
         &= \frac{(1-p)2g^{t+\epsilon}}{1-pg} + n. 
\end{align*}
Similarly we have
\begin{align*}
    L(t-1) &= (1-p)(2g^{t-1+\epsilon} + R(t)) + (1-p)p(2g^{t+\epsilon} + R(t+1)) + \\
           &= (1-p)2g^{t-1+\epsilon}\sum_{i=0}^{\infty}(pg)^i + (1-p)\sum_{i=0}^{\infty}p^i R(t+i). 
\end{align*}
From the above we derive a general expression of $L(k)$ for some $k<t+1$ as follows
\begin{align*}
    L(k) &= (1-p)(2g^{k+\epsilon} + R(k+1)) + (1-p)p(2g^{k+1+\epsilon} + R(k+2)) + \ldots \\
           &= (1-p)\sum_{i=0}^{t-2-k}(2g^{k+\epsilon+i} + R(k+1+i)) + (1-p)\sum_{i=t-1-k}^{\infty}p^i (2g^{k+\epsilon+i} + R(k+1+i)) \\
           &= (1-p)2g^{k+\epsilon}\sum_{i=0}^{\infty}(pg)^i + (1-p)\sum_{i=0}^{t-2-k}R(k+1+i) + (1-p)\sum_{i=t-1-k}^{\infty}p^i n. 
\end{align*}
Hence we conclude that 
\begin{equation}\label{eq: expected time left second case}
    L(k) = (1-p)2g^{k+\epsilon}\sum_{i=0}^{\infty}(pg)^i + (1-p)\sum_{i=0}^{t-2-k}R(k+1+i) + (1-p)\sum_{i=t-1-k}^{\infty}p^i n.
\end{equation}
Next we set, for $k=0,\ldots,t$,
\begin{align}
    \alpha_k :=& np^{t-k} + (1-p)2g^{k+\epsilon}\sum_{i=0}^{t-1-k}(pg)^i + (1-p)p^{t-1-k}L(t) \notag\\
              =& np^{t-k} - np^{t-k} + np^{t-1-k} + \frac{(1-p)2g^{k+\epsilon}(1-(pg)^{t-k})}{1-pg} + \frac{2g^{t+\epsilon}(1-p)^2p^{t-1-k}}{1-pg} \notag\\
              =& \frac{2g^{t+\epsilon}(1-p)^2p^{t-1-k}}{1-pg} + \frac{(1-p)2g^{k+\epsilon}}{1-pg}(1-(pg)^{t-k}) + np^{t-1-k} \notag\\
              =& \frac{2g^{k+\epsilon}(1-p)}{1-pg}\left((1-p)p^{t-1-k}g^{t-k}+1-(pg)^{t-k}\right) + np^{t-1-k} \notag\\
              =& \frac{2g^{k+\epsilon}(1-p)}{1-pg}\left(g(pg)^{t-1-k}-(pg)^{t-k}+1-(pg)^{t-k}\right) + np^{t-1-k} \notag\\
              =& \frac{2(1-p)g^{k+\epsilon}}{1-pg}\left(1+(1-2p)g(pg)^{t-1-k}\right) + np^{t-1-k}\label{eq: alpha plus epsilon second case} \\
    \beta_k :=& (1-p)2g^{k+\epsilon} \sum_{i=0}^\infty (pg)^i 
    + (1-p) \sum_{i=t-1-k}^\infty p^i n \notag \\
             =& \frac{2 (1-p) g^{k+\epsilon}}{1-pg}+n p^{t-1-k}, \label{eq: beta plus epsilon second case}
\end{align}
as well as we also set $\alpha^T=(\alpha_0, \ldots, \alpha_t), \beta^T=(\beta_0, \ldots, \beta_t)$.

As before from equations~\eqref{eq: expected time right second case}, \ref{eq: expected time left second case}, \ref{eq: alpha plus epsilon second case}, and \ref{eq: beta plus epsilon second case}, we have that $R(0), \ldots, R(t-1)$ and $L(0), \ldots, L(t-1)$ satisfy the following linear system of $2t$ constraints
\begin{align}
R(k) - (1-p)
\sum_{i=0}^{t-2-k}p^i L(k+1+i)
=& \alpha_k, ~~k=0,\ldots,t-1, \label{eq: right for left second case}\\
L(k) - (1-p)
\sum_{i=0}^{t-2-k} p^i R(k+1+i)  
=& \beta_k, ~~k=0,\ldots,t-1. \label{eq: left for right second case}
\end{align}
Solving the system we compute $R(0)$ as follows.
\begin{equation} \label{eq:r zero bigger epsilon}
    R(0) = n + \frac{g^\epsilon (-1 + p) (-2 + 
    g (2 p + 
       g^{t-1} (1 + g - 2 g p - (-1 + g) (-1 + 2 p)^{t})))}{(-1 + 
    g) (-1 + g p)}
\end{equation}
If we replace $t$ with $t-1$ in \eqref{eq:r zero smaller epsilon} we derive \eqref{eq:r zero bigger epsilon}. Additionally, because this is the case of $\epsilon \geq \delta$,  we compute
$$\int_{\delta}^{1} R(0) \, d\epsilon$$
which simplifies to
\begin{equation} \label{eq:integral of r zero bigger epsilon}
    n - n \delta + \frac{(g - g^\delta) (-1 + p) (-2 + g^t + g^{1 + t} + 
    2 g p - 2 g^{1 + t} p - (-1 + g) g^t (-1 + 2 p)^t)}{(-1 + 
    g) (-1 + g p) \ln g}.
\end{equation}

\subsection{Combining the two cases}

Now we are ready to combine the two cases in order to compute the total expected termination time. More specifically the we have that
$$R(0) = \int_{0}^{\delta} R_{\epsilon < \delta}(0) + \int_{\delta}^{1} R_{\epsilon \geq \delta}(0).$$
Using \eqref{eq:integral of r zero smaller epsilon} and \eqref{eq:integral of r zero bigger epsilon}, we obtain 
\begin{multline}
    n \left(1 + \frac{(1 - p) (-1 + g (-1 + 2 p)) (1 + (-1 + 2 p)^t)}{(-1 + 
    g p) \ln g}\right) 
    - \frac{2 (-1 + p)}{(-1 + g p) \ln g} 
    + \frac{ 2 g (-1 + p) (p + g^t (-1 + p) (-1 + 2 p)^t)}{(-1 + g p) \ln g}.
\end{multline}

Scaling by the target distance $n$ gives the promised competitive ratio
\begin{multline}
    1 + \frac{(1 - p) (-1 + g (-1 + 2 p)) (1 + (-1 + 2 p)^t)}{(-1 + 
    g p) \ln g} 
    - \frac{2 (-1 + p)}{n(-1 + g p) \ln g} 
    + \frac{ 2 g (-1 + p) (p + g^t (-1 + p) (-1 + 2 p)^t)}{n(-1 + g p) \ln g}
\end{multline}
where $n \in (g^t, g^{t+1}]$.

\ignore{
Moreover, if we require that the target is placed arbitrarily far from the origin, which actually implies that $t \to \infty$ we compute the following limits
\begin{align*}
    &\lim_{t \to \infty} - \frac{2 (-1 + p)}{g^{t+\delta}(-1 + g p) \ln g} = 0 \\
    &\lim_{t \to \infty} \frac{ 2 g (-1 + p) (p + g^t (-1 + p) (-1 + 2 p)^t)}{g^{t+\delta}(-1 + g p) \ln g} = 0 \\
    &\lim_{t \to \infty} \frac{(1 - p) (-1 + g (-1 + 2 p)) (1 + (-1 + 2 p)^t)}{(-1 + 
    g p) \ln g} = \frac{(-1 + p) (-1 + g (-1 + 2 p))}{(1 - g p) \ln g}
\end{align*}
which finally gives us the competitive ratio of
\begin{equation}
    1 + \frac{(-1 + p) (-1 + g (-1 + 2 p))}{(1 - g p) \ln g}
\end{equation}
}

\section{Proof of Theorem~\ref{thm: 2 agents improved competitive}}
\label{sec: proof of thm: 2 agents improved competitive}

Here we give all details for the proof of Theorem~\ref{thm: 2 agents improved competitive}.

\begin{proof}[of Theorem~\ref{thm: 2 agents improved competitive}]

For the sake of exposition, we present again Algorithm~\ref{algo: search with two wireless} and proceed with its analysis. 

\textit{Phase 1:} First the two $p$-faulty agents start searching in opposite direction at full speed until the target is found (since this is the initial move, there can be no mistake in choosing the direction). Without loss of generality we assume that the target is reported by the finder at time $1$, when it is also communicated to the non-finder instantaneously. 

\textit{Phase 2:} At this moment the non-finder attempts to turn (possibly without success), and switches its speed to $s=s(p)<1$ that will be determined later. 
Now we turn our attention to the target finder. The finder becomes the Leader, and chooses $\gamma=\gamma(p)>0$. Consequently, the finder passes the target for time $\gamma$, and then runs Algorithm~\ref{algo:subroutine change direction} with the help of the target playing the role of the immobile Follower, until the target is reached again. 

\textit{Phase 3:} The finder continues in the same direction until it meets the non-finder. Note that this meeting is realized with probability $1$ independently of whether the non-finder turned successfully or not, since the non-finder's speed is now $s<1$. 

\textit{Phase 4:} Once the meeting is realized, the non-finder assumes the role of the Follower, and the finder remains the Leader. The Follower stays put, and the Leader continues moving in the same directoin (away from the target and the Follower) for time $\gamma$ and then runs Algorithm~\ref{algo:subroutine change direction} with the Follower being the immobile agent, in order to turn. 

\textit{Phase 5:} Once the turn is realized, the two $p$-faulty agents move to the target. 

Next we compute the expected termination time.
Denote by $W$ the random variable that equals the time between two consecutive visitations by the finder, as part of Phase 2. 
Let also $T$ be the random variable that equals the time that the two agents meet for the second time, as part of Phase 4. 


Back to Phase 2, we look at the independent event concerning the non-finder. 
On the one hand with probability $1-p$ the non-finder successfully turns, so when the finder reaches the target for the second time and starts moving towards the non-finder, the two are at distance $2-sW$ and have relative speed $1+s$. This means that they meet in additional time $(2-sW)/(1+s)$, during which time the non-finder has moved closer to the target by $s(2-sW)/(1+s)$. Hence, when the two agents meet, they are at distance $2-sW-s(2-sW)/(1+s)$ from the target. 

On the other hand with probability $p$ the non-finder fails to turn, so when the finder reaches the target for the second time and starts moving towards the non-finder, the two are at distance $2+sW$ and have relative speed $1-s$ (because the move towards the same direction). This means that they meet in additional time $(2+sW)/(1-s)$, during which time the non-finder has moved further from the target by $s(2+sW)/(1-s)$. Hence, when the two agents meet, they are at distance $2+sW+s(2+sW)/(1-s)$ from the target. 

But then, by the Linearity of expectation, the termination time (which is equal to the competitive ratio) equals
\begin{align*}
& 1+ \ex[W] + \ex[T]
+
p 
\left(
\ex[(2+sW)/(1-s) + 2+sW+s(2+sW)/(1-s)]
\right)\\
&~~+
(1-p) 
\left(
\ex[(2-sW)/(1+s) + 2-sW-s(2-sW)/(1+s)]
\right)
\\
=&
\frac{5-s (s+4-8 p)}{1-s^2}
+\ex[T]
+
\frac{s (4 p+s-2)+1}{1-s^2}
\ex[W]
\end{align*}

We choose $s=s(p) = \frac{1-2 \sqrt{p-p^2}}{1-2 p}$ that can be easily seen to attain values in $(0,1)$ for all $p \in (0,1/2)$, hence it is a valid choice for a speed.\footnote{The value of $s$ is chosen as it corresponds to the minimizer of $\frac{5-s (s+4-8 p)}{1-s^2}$ over all parameters $0<p<1/2$.}
Now we substitute back to the formula of the competitive ratio, and after we simplify algebraically, the expression becomes
$$
3+4 \sqrt{(1-p) p} + \ex[T] + 2\sqrt{(1-p) p}\ex[W].
$$
The main claim follows by remembering that both $\ex[T],\ex[W]$ are proportional to $\gamma$, and hence can be made arbitrarily small.
\qed \end{proof}

\end{document}